\documentclass[journal, draftcls, onecolumn]{IEEEtran}

\usepackage{amsmath, mathrsfs, graphicx,epsfig,amssymb, color,amsfonts}
\usepackage[center]{caption}

\newtheorem{lemma}{Lemma}

\newtheorem{thm}{Theorem}

\newtheorem{rem}{Remark}

\newtheorem{ex}{Example}

\begin{document}

\title{The diversity-multiplexing tradeoff of the MIMO Z interference channel}

\author{{\Large Sanjay ~Karmakar ~~~~~~Mahesh~ K. ~Varanasi}
\thanks{The authors are with the Department
of Electrical Computer and Energy Engineering, University of Colorado at Boulder, Boulder,
CO, 30809 USA e-mail: (sanjay.karmakar@colorado.edu, varanasi@colorado.edu). This work was presented in part at the  2010 IEEE International Symposium on Information Theory (ISIT), Jun. 2010, Austin, TX, USA \cite{Karmakar-Varanasi-Z-IC-ISIT}}
}

%

\maketitle

\begin{abstract}
The fundamental generalized diversity-multiplexing tradeoff (GDMT) of the quasi-static fading MIMO Z interference channel (Z-IC)  is established for the general Z-IC with an arbitrary number of antennas at each node under the assumptions of full channel state information at the transmitters (CSIT) and a short-term average power constraint. In the GDMT framework, the direct link signal-to-noise ratios (SNR) and cross-link interference-to-noise ratio (INR) are allowed to vary so that their ratios relative to a nominal SNR in the dB scale, i.e., the SNR/INR exponents, are fixed. It is shown that a simple Han-Kobayashi message-splitting/partial interference decoding scheme that uses only partial CSIT -- in which the second transmitter's signal depends only on its cross-link channel matrix and the first user's transmit signal doesn't need any CSIT whatsoever -- can achieve the full-CSIT GDMT of the MIMO Z-IC. The GDMT of the MIMO Z-IC under the No-CSIT assumption is also obtained for some range of multiplexing gains. 
The size of this range depends on the numbers of antennas at the four nodes and the SNR and INR exponents of the direct and cross links, respectively. For certain classes of channels including those in which the interfered receiver has more antennas than do the other nodes, or when the INR exponent is greater than a certain threshold, the GDMT of the MIMO Z-IC under the No-CSIT assumption is completely characterized. 
\end{abstract}

%

\IEEEpeerreviewmaketitle

\newpage

\section{Introduction}
An information theoretic consideration of a multiuser wireless network begins by allowing the transmitters and receivers to function in arbitrarily complex ways and proceeds to narrow down and perhaps even identify specific communication schemes that are optimal in a capacity or an approximated capacity sense. The engineering of wireless networks has however proceeded to a large extent on a common sense approach of restricting different groups of users to operate in weakly interfering or disjoint signal spaces, as in cellular networks. This has, at least in part, lead information theorists to focus considerable efforts on investigating the fundamental limits of point-to-point links and multiple-access (many-to-one) and broadcast (one-to-many) sub-networks as building blocks of larger wireless networks. However, new communication schemes that have emerged out of recent information theoretic research on multiple uni/multi-cast interference networks -- such as message splitting/partial interference decoding (cf. \cite{Han_Kobayashi,CMG,ETW1,TT,Sanjay_Varanasi_constant_gap_to_capacity}) and interference alignment (cf. \cite{MMK-IA,Cadambe-Jafar-IA,Vaze_Varanasi_MIMOIC_Delayed,Wang_Gou_Jafar}) -- suggest that the idea of limiting {\em a priori} the design of wireless networks by partitioning them into only those sub-networks may be highly suboptimal in terms of throughput performance, and must therefore be revisited.


In this paper, we study a two uni-cast network model known as the one-sided interference channel or the Z-interference channel (Z-IC). In the Z-IC, two transmitters communicate to their corresponding receivers in the same signal space but only one of the transmitters interferes with its unpaired receiver (see Fig.~\ref{channel_model_Z_IC}). Indeed, the Z-IC is the simplest interference network model that simultaneously contains the two important phenomenon in wireless networks that occur separately in multiple-access and broadcast networks, namely {\it superposition} and {\it broadcast}, respectively. Besides being a basic building block of general interference networks, Z-ICs are also the natural information theoretic model for various practical wireless communication scenarios such as femto-cells~\cite{Femto} or a line network of four nodes
where the two transmitters $(Tx_1,~Tx_2)$ and their corresponding receivers $(Rx_1,~Rx_2)$ are interconnected by the links described by using the arrows in $Tx_1~\rightarrow~ Rx_1~\leftarrow ~Tx_2~\rightarrow~ Rx_2$. Moreover, the Z-IC is a special case of the $2$-user interference channel (IC) where both transmitters interfere at their respective unpaired receivers. Thus optimal coding and decoding schemes (with optimality defined in an exact or some approximated capacity sense) on a Z-IC may reveal useful insights into the $2$-user IC as well. 


In spite of considerable research on the Z-IC, including the important cases of the Gaussian SISO (single-input, single-output) or MIMO Z-IC, its capacity region remains unknown in general. The capacity region has however been found for various special cases with restricted values of the channel parameters in  \cite{Sato78,Sato,Costa85,Sason,Liu-Ulukus,SCKP,Motahari_Khandani,Liu_Goldsmith}. Moreover, all of those results on exact capacity are derived under the assumption that the channel coefficients are time-invariant. More recently, the characterization of the capacity region of the ergodic fading SISO Z-IC was given to within a bounded gap of 12.8 bits in \cite{Zhu-Guo-ZIC} independently of the SNR and fading statistics in the ``No-CSIT" scenario wherein the transmitters have knowledge of the statistics of of the fading channel coefficients but not their realizations. However, the central enabling idea of this result, namely that of approximating the fading channel by time-varying deterministic channels, otherwise known as the layered erasure model, can not be extended to the MIMO case (see Section II-E of \cite{Avestimehr-net-inf-flow} for more details on this point). At the same time, the enormous potential throughput/reliability benefits of employing multiple antenna nodes cannot be ignored in many present and most future wireless networks. These considerations motivate the study of the fading MIMO Z-IC in this work under the assumptions of full CSIT as well as No-CSIT. Under the No-CSIT assumption the fast fading MIMO IC has been studied recently in terms of the generalized degrees of freedom (GDoF) in \cite{Chinmay_Sanjay_Varanasi_isit2011} and in terms of the degrees of freedom (DoF) in \cite{Ke-Wang} when the transmitters have reconfigurable antennas. While the DoF (GDoF) metrics are important in that they capture the simultaneously available time/frequency/space (and signal-level) dimensions in a network, and hence characterize the so-called ``pre-log" factors in the manner in which capacity regions scale with increasing SNR, they do not reveal any information about the other crucial aspect of practical networks, namely, {\em reliability} of communication, which traditionally would be revealed for fixed channels through the characterization of (or bounds on) error exponents (cf. \cite{GallagerRG:info,GallagerRG:multi,GuessT:errorExp}) for rate-tuples within the capacity region.



In this paper, we study the MIMO Z-IC under quasi-static Rayleigh fading in which codeword lengths are comparable to the coherence time of the channel, and hence experience a single fading realization, in the limit of large block length. Since any non-zero rate pair, no matter how small, cannot be supported by some set of fading realizations with positive probability, the Shannon theoretic capacity region of this network is meaningless. An outage view is therefore adopted in which the probability of the set of channel realizations that cannot support a given rate pair (no matter what admissible communication scheme is used) is called the outage probability. Clearly, there is a tradeoff between the achievable rate and outage probability. The ``larger" the rate pair, the larger is the set of channel realizations for which that rate pair cannot be reliably supported, and hence larger the outage probability. The well-known diversity-multiplexing tradeoff (DMT) -- introduced in the seminal work of Zheng and Tse in the context of a point-to-point MIMO channel \cite{tse1} 
-- succinctly captures this tradeoff in the high SNR regime. However, unlike the DMT framework in a point-to-point channel where the communication link can be characterized by a single SNR, in a multiuser setting such as the MIMO Z-IC, there are multiple links (two direct links and one cross-link in the MIMO Z-IC) and it is more the rule rather than the exception in practical scenarios that signals and interferences are received at a terminal at strengths that can be significantly different. Consequently, there is much additional insight to be gained by allowing the two direct-link SNRs and the cross-link interference-to-noise ratio (INR) to vary with a nominal SNR, denoted as $\rho$, in such a manner that their ratio relative to the nominal SNR {\em in the dB scale} are taken to be some fixed but arbitrary numbers (these are henceforth called SNR/INR exponents). This important idea was introduced in the DoF study of the two-user time-invariant SISO IC in \cite{ETW1} resulting in its so-called Generalized DoF (GDoF) region which was subsequently generalized by the authors of this paper to the two-user time-invariant MIMO IC in \cite{Karmakar-Varanasi-GDoF}. The works of \cite{ArPv} and \cite{ABarxiv} later adopted the model of scaling SNRs and INRs in the manner of \cite{ETW1} within the DMT framework for quasi-static fading SISO ICs. The resulting DMT which is a function of not only the multiplexing gains but also the SNR and INR exponents is referred to as the generalized DMT (GDMT) to distinguish it from the usual DMT (used for example in the case of the MIMO multiple-access channel (MAC) in \cite{TseD:MAC-DMT,PrasadN:VB:Aller03}) in which all SNR/INR exponents are set equal to unity.

Other than the conference version of this work \cite{Karmakar-Varanasi-Z-IC-ISIT}, prior results on the DMT of the Z-IC can be divided into two classes depending on the assumption on the availability of channel state information at the transmitters (CSIT) or its complete lack thereof (No-CSIT). In \cite{ABK}, the authors derive the achievable No-CSIT DMT on a SISO Z-IC of a scheme where the transmitters use Gaussian codes while on the receiver side interference is treated as noise at the interfered receiver (ITN). In that work, the outage event of only the interfered receiver is analyzed. 
More recently in \cite{Nafea-Seddik-Nafie-Gamal-ZIC}, the achievable DMT of the Han-Kobayashi (HK) coding scheme \cite{Han_Kobayashi} on a SISO Z-IC was characterized assuming fixed power splitting among the private and public messages at the interfering (second) transmitter while considering the general case in which the two receivers are allowed to achieve different diversity orders for any targeted multiplexing-gain pair. In \cite{Nafea-Seddik-Nafie-Gamal-ZIC-ARQ}, this result was extended to the Z-IC with automatic repeat request (ARQ). It should be emphasized that all of the above mentioned results present only lower bounds to the fundamental DMT of the channel, the characterization of which even for the SISO Z-IC is a wide open problem. In contrast to the No-CSIT case, there is no direct result available in the literature on the DMT of the Z-IC with CSIT. However, the DMT of a 2-user SISO IC with CSIT obtained in \cite{EaOl} can also be achieved on the corresponding Z-IC and hence serves as a lower bound to the DMT of the SISO Z-IC with CSIT. This is because the Z-IC is more capable than the corresponding two-user IC because the second receiver in the former does not encounter any interference. In summary, the Z-IC with CSIT has not been studied previously from the DMT perspective while only lower bounds to the DMT of Z-ICs are available in the No-CSIT case, and that too just in the SISO case.

In this work (and in \cite{Karmakar-Varanasi-Z-IC-ISIT}), we completely characterize the GDMT of the MIMO Z-IC with CSIT in its most general form, i.e., with arbitrary number of antennas at each node and with arbitrary SNR/INR exponents. On the No-CSIT front, our result improves upon and generalizes the previous results in the following two directions: 1) in contrast to only an achievable DMT (i.e., a lower bound on the DMT of the channel) of a particular coding scheme derived in earlier works, we characterize here the fundamental DMT itself, i.e., the best achievable DMT among all possible coding schemes and 2) we consider the MIMO Z-IC with an arbitrary number of antennas at each node, a particular (single-antenna) case of which of course establishes the fundamental GDMT of the SISO Z-IC. For instance, in Example~\ref{ex-siso-strong-no-csit-zic} we obtain the fundamental DMT of a class of Z-ICs restricted to the SISO case. This GDMT coincides with the achievable DMT of the {\it common message only} (CMO) coding scheme in \cite{Nafea-Seddik-Nafie-Gamal-ZIC}. Even in this special case, our result provides the fundamental GDMT, and hence both achievability and converse results, whereas \cite{Nafea-Seddik-Nafie-Gamal-ZIC} provides only an achievability result. Moreover, Example \ref{ex-siso-strong-no-csit-zic} is a special case of a more general result given in the conference version of this work in \cite{Karmakar-Varanasi-Z-IC-ISIT} (see Lemma 5 therein) which precedes \cite{Nafea-Seddik-Nafie-Gamal-ZIC}.


As a stepping stone to characterizing the GDMT we first derive inner and outer bound regions to the instantaneous capacity region of the MIMO Z-IC. These regions are defined in terms of three bounds each on the two individual rates and one sum rate. The corresponding bounds on these rates of the inner and outer regions are shown to be within a constant gap of each other (independently of all channel parameters), therefore characterizing the instantaneous capacity within a bounded gap. This result is obtained in a straightforward manner from our bounded gap result for the 2-user MIMO IC in \cite{Sanjay_Varanasi_constant_gap_to_capacity}. The achievability of the inner region is based on a simple coding scheme chosen from the family of Han-Kobayashi coding schemes \cite{Han_Kobayashi}, where the signal to be transmitted by the second (interfering) user depends only on its channel matrix to the first (interfered) receiver, whereas the transmitted signal of the first transmitter does not use any CSIT. 

The analysis of the outage probabilities formulated based on the achievable rate region and the outer bound region yield a lower and an upper bound to the GDMT of the MIMO Z-IC, respectively. Since the achievable rate region and the outer region are within a constant gap to each other, and in the high SNR regime a constant gap is insignificant, the two GDMT bounds are identical. Therefore, the GDMT of the above mentioned HK coding scheme actually represents the fundamental GDMT of the channel. However, in contrast to the GDMT of the SISO case~\cite{ABK},~\cite{Nafea-Seddik-Nafie-Gamal-ZIC} which requires the joint statistics of three mutually independent exponentially distributed {\em scalar} random variables, the characterization of the GDMT of the MIMO Z-IC involves the joint eigenvalue distribution of three mutually correlated random Wishart matrices. The characterization of such joint eigenvalue statistics in general is a challenging problem. A similar problem was recently solved by the authors in \cite{DDF-DMT-Journal-IT} in establishing the DMT of the MIMO relay channel. Using the result therein, we first derive the joint distribution that is relevant to the problem of determining the GDMT of the MIMO Z-IC. This distribution result allows us to characterize the GDMT of the MIMO Z-IC with CSIT as the solution of a convex optimization problem. While such an optimization problem can be solved using numerical methods, closed-form solutions are also derived for several special cases for deeper insight. 

Next, we drop the assumption of CSIT and characterize the achievable GDMT of a transmission scheme which does not utilize any CSI whatsoever at the transmitters. Clearly, this gives a lower bound to the No-CSIT DMT of the channel. On the other hand, the full-CSIT GDMT of the Z-IC serves as an upper bound to the No-CSIT DMT of the channel. Using these lower and upper bounds, we characterize the fundamental GDMT under the NO-CSIT assumption of two classes of MIMO Z-ICs for which they coincide. The first class consists of MIMO Z-ICs with an equal number of antennas at all four nodes and an SNR exponent of the cross-link that is larger than a certain threshold (e.g.,Theorem~\ref{thm:DMT-symmetric-no-CSIT}) and the second class consists of MIMO Z-ICs in which the number of antennas at the interfered receiver is larger than a certain threshold (e.g., Theorem~\ref{thm:larger_antenna}). To the best of our knowledge, this paper provides the first characterizations of the fundamental GDMT on MIMO Z-ICs (including the SISO case) under both CSIT and No-CSIT scenarios.

The rest of the paper is organized as follows. In Section~\ref{sec_channel_model}, we provide a description of the MIMO Z-IC model and the definition of the GDMT framework in Section \ref{subsec:def-dmt}. In Section~\ref{subsec:ZIC-approximate-capacity}, we derive inner and outer bounds to the instantaneous capacity regions of the channel which is then used in Section~\ref{sec_DMT_Z_channel} to establish its GDMT as a solution of a convex optimization problem. In Section \ref{subsec:ZIC-analytic-solutions}, we derive analytic solutions to this general optimization problem and hence the GDMT of various special classes of MIMO Z-ICs. In Section \ref{sec:ZIC-no-CSIT}, we characterize the GDMT under the no-CSIT assumption and conclude with Section~\ref{sec:conclusion}. Some of the proofs are relegated to the appendix to enable a clearer exposition of the main results of this paper.

{\bf Notations:}
We denote the conjugate transpose of the matrix $A$ as $A^{\dagger}$ and its determinant as $|A|$. $\mathbb{C}$ and $\mathbb{R}$ represent the field of complex and real numbers, respectively. The set of real numbers $\{x\in \mathbb{R}: a\leq x\leq b\}$ will be denoted by $[a,b]$. Furthermore, $(x\land y)$, $(x\vee y)$ and $(x)^+$ represent the minimum of $x$ and $y$, the maximum of $x$ and $y$, and the maximum of $x$ and $0$, respectively. All the logarithms in this paper are with base $2$. We denote the distribution of a complex circularly symmetric Gaussian random vector with zero mean and covariance matrix $Q$ as $\mathcal{CN}(0,Q)$. Any two functions $f(\rho)$ and $g(\rho)$ of $\rho$, where $\rho$ is the signal to noise ratio (SNR) defined later, are said to be exponentially equal and denoted as $f(\rho)\dot{=}g(\rho)$ if,
$\lim_{\rho \to \infty} ~\frac{\log(f(\rho))}{\log(\rho)} = \lim_{\rho \to \infty} ~\frac{\log(g(\rho))}{\log(\rho)}$. The operations $\dot{\geq}$ and $\dot{\leq}$ are defined similarly.

\section{Channel Model and Preliminaries}
\label{sec_channel_model}
We consider a MIMO Z-IC as shown in Fig. \ref{channel_model_Z_IC}, where user $1$ ($Tx_1$) and user $2$ ($Tx_2$) have $M_1$ and $M_2$ antennas and receiver $1$ ($Rx_1$) and $2$ ($Rx_2$) have $N_1$ and $N_2$ antennas, respectively. This channel will be referred to henceforth as the $(M_1,N_1,M_2,N_2)$ Z-IC. A slow fading Rayleigh distributed channel model is considered where ${H}_{ij}\in \mathbb{C}^{N_j\times M_i}$ represents the channel matrix between $Tx_i$ and $Rx_j$ and it is assumed that $H_{11}$, $H_{21}$ and $H_{22}$ are mutually independent and contain mutually independent and identically distributed (i.i.d.) $\mathcal{CN}(0,1)$ entries. These channel matrices remain fixed for a particular fade duration of the channel and change in an i.i.d. fashion in the next. Perfect channel state information is assumed at both the receivers (CSIR) and both the transmitters (CSIT) at first and following that the No-CSIT case is considered where there is only CSIR. At time $t$, $Tx_i$ chooses a vector ${X}_{it}\in \mathbb{C}^{M_i\times 1}$ and sends $\sqrt{P_i}{X}_{it}$ over the channel, where the input signals are assumed to satisfy the following short term average power constraint:

\begin{IEEEeqnarray}{c}
\label{power_constraint}
\frac{1}{N}\sum_{t=1+N(k-1)}^{N+N(k-1)}\textrm{tr}(Q_{it}) \leq 1, \forall ~k\geq 1,~ i=1,2,~\textrm{where}~Q_{it}=\mathbb{E}\left({X}_{it}{X}_{it}^{\dagger}\right),
\end{IEEEeqnarray}
where $N$ represents the fade duration or coherence time of the network in which the channel matrices remain fixed.
\begin{rem}
Note that since the transmitters are not allowed to allocate power across different fades of the channel, the channel is in the outage setting, i.e., the delay limited capacity of each of the links of the Z-IC is zero~\cite{HanlyTse}. Thus, the MIMO Z-IC under the short-term power constraint is outage limited and its DMT is meaningful.
\end{rem}

\begin{figure}[!thb]
\setlength{\unitlength}{1mm}
\begin{picture}(80,40)
\thicklines
\put(60,10){\vector(1,0){40}}
\put(60,32){\vector(1,0){40}}
\put(60,11){\vector(2,1){40}}
\put(50,9){$\mathbf{(M_2)}$}
\put(109,9){$\mathbf{(N_2)}$}
\put(50,31){$\mathbf{(M_1)}$}
\put(109,31){$\mathbf{(N_1)}$}

\put(43,9){$\mathbf{Tx_2}$}
\put(102,9){$\mathbf{Rx_2}$}
\put(43,31){$\mathbf{Tx_1}$}
\put(102,31){$\mathbf{Rx_1}$}

\put(78,12){$\mathbf{{H}_{22}}$}
\put(78,34){$\mathbf{{H}_{11}}$}
\put(62,17){$\mathbf{{H}_{21}}$}
\end{picture}
\caption{Channel model for the Z-IC.}
\label{channel_model_Z_IC}
\end{figure}
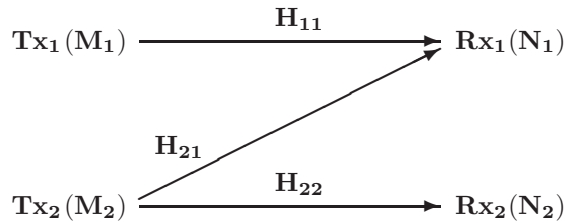

The received signals at time $t$ can be written as
\begin{IEEEeqnarray*}{l}
\label{system_eq_two_user_IC1}
Y_{1t}=\eta_{11} \sqrt{P_1} {H}_{11}{X}_{1t}+ \eta_{21} \sqrt{P_2} {H}_{21}{X}_{2t}+Z_{1t},\\
Y_{2t}= \eta_{22} \sqrt{P_2} {H}_{22}{X}_{2t}+Z_{2t},
\end{IEEEeqnarray*}
where $Z_{it}\in\mathbb{C}^{N_i\times 1}$ are i.i.d as $\mathcal{CN}(\mathbf{0}, I_{N_i})$ across $i$ and $t$ and $\eta_{ij}$ represents the signal attenuation factor~\cite{ABarxiv} from $Tx_i$ to $Rx_j$. The above equations can be equivalently written as
\begin{IEEEeqnarray}{l}
\label{system_eq_two_user_IC2}
Y_{1t}=\sqrt{\textrm{SNR}_{11}}H_{11}\hat{X}_{1t}+ \sqrt{\textrm{INR}_{21}}H_{21}\hat{X}_{2t}+Z_{1t};\\
Y_{2t}= \sqrt{\textrm{SNR}_{22}}H_{22}\hat{X}_{2t}+Z_{2t},
\end{IEEEeqnarray}
where the normalized inputs $\hat{X}_i$s satisfy equation \eqref{power_constraint} with equality and $\textrm{SNR}_{ii}$ and $\textrm{INR}_{ji}$ are the signal-to-noise ratio and interference-to-noise ratio, respectively, at receiver $i$, denoted henceforth as $\rho_{ii}$ and $\rho_{ji}$, respectively. In practice the SNRs and INRs of the various links in a wireless network will in general be arbitrarily disparate. Further, since the difference in performance due to those disparities cannot be captured through the DMT metric if the signal strengths differ only by say a constant factor, we let the different SNRs and INRs to vary exponentially with respect to a nominal SNR, denoted as $\rho$, with different {\em SNR/INR exponents} as follows:
\begin{IEEEeqnarray}{l}
\alpha_{11}=\frac{\log (\textrm{SNR}_{11})}{\log(\rho)}, \quad \alpha_{22}=\frac{\log(\textrm{SNR}_{22})}{\log(\rho)}, \quad \alpha_{21}=\frac{\log(\textrm{INR}_{21})}{\log(\rho)}.
\end{IEEEeqnarray}
For brevity, in the sequel we shall use the following notations: $\mathcal{H}=\{H_{11},H_{21},H_{22}\}$, $\bar{\rho}=[\rho_{11},\rho_{21},\rho_{22}]$ and $\bar{\alpha}=[\alpha_{11},\alpha_{21},\alpha_{22}]$. 


The DMT framework is defined carefully in \cite{tse1} for the MIMO link and for the MIMO MAC in \cite{TseD:MAC-DMT,PrasadN:VB:Aller03}. In what follows, we define the GDMT framework for the MIMO Z-IC that includes the SNR/INR exponents.

\subsection{The GDMT of the MIMO Z-IC}
\label{subsec:def-dmt}
Consider a coding scheme consisting of a family of codebook pairs $\{(\mathcal{C}_1(\rho), \mathcal{C}_2(\rho))\}$, parameterized by $\rho $, with $\mathcal{C}_i(\rho)$, the codebook for the $i^{th}$ transmitter, having $2^{LR_i(\rho)}$ codewords with block length $L$ and rate $R_i(\rho)$. The multiplexing gain pair of this coding scheme, denoted as $(r_1, r_2)$, is defined through
\begin{IEEEeqnarray}{l}
\label{eq:def-mg}
r_i=\lim_{\rho\to \infty} \frac{R_i(\rho)}{\log(\rho)} \, ,~\textrm{for }~i=1,2.
\end{IEEEeqnarray}

\begin{rem}
The maximum asymptotic rate supportable by the first direct link (when $Tx_2$ is silent) is $R_1^{\max} \approx \min\{M_1,N_1\}\log(\rho_{11})$. Hence
$ r_1^{\max}=\min\{M_1,N_1\}\alpha_{11} $, so that if $\alpha_{11}>1 $, it is strictly greater than $ \min\{M_1,N_1\} $.
That the direct link can support a multiplexing-gain strictly greater than $\min\{M_1,N_1\}$ may seem like a contradiction but this is only a consequence of the multiplexing gains being defined in \eqref{eq:def-mg} with respect to (w.r.t.) the nominal SNR $\rho$. With respect to the direct link's SNR $\rho_{11}$ the maximum multiplexing-gain is still $\min\{M_1,N_1\}$, irrespective of the value of $\alpha_{11}$. 
Alternatively, we can set $\rho=\rho_{11}$ without loss of generality, which amounts to setting $\alpha_{11}=1$.
\end{rem}

Now, to define the metric of primary importance in this work, namely the GDMT, let $\mathcal{P}_{e,\mathcal{C}}(\bar{\rho},r_1,r_2)$ represent the larger of the two average probability of errors at the two receivers (averaged over the additive Gaussian noise, all codewords of the codebook pair and over the randomness of the channel) at a multiplexing gain pair $(r_1,r_2)$ and SNR of $\rho$. Let $\mathcal{P}_{e}^*(\bar{\rho},r_1,r_2)$ represent the infimum of $\mathcal{P}_{e,\mathcal{C}}(\bar{\rho},r_1,r_2)$ among all possible admissible coding schemes (that satisfy the short-term power constraint).
Then the fundamental GDMT of the MIMO Z-IC is given as
\begin{IEEEeqnarray}{l}
d_{\textrm{Z-IC}}^*(r_1,r_2)=\lim_{\rho \to \infty} \frac{-\log\left(\mathcal{P}^*_e(\bar{\rho},r_1,r_2)\right)}{\log(\rho)}.
\end{IEEEeqnarray}
Note that $d_{\textrm{Z-IC}}^*(r_1,r_2)$ is a function of the relative scaling parameters of the different links $\bar{\alpha}$. However, for brevity of notation, we do not make this dependence explicit.

A general approach to characterizing the GDMT of a communication network, whose exact maximum instantaneous end-to-end mutual information (MIMI) region -- namely, the capacity region of the network given any fixed channel realization) -- is not known, is to find inner and outer bounds to it given that fixed channel realization. Then, from these bounds, lower and upper bounds on the probability of an appropriately defined outage event are derived, respectively. If these bounds have identical exponents (relative to SNR) then that exponent is the GDMT of the network. In the following section, we specify precisely such inner and outer bounds on the MIMI region.

\subsection{Inner and Outer Bounds to the MIMI Region}
\label{subsec:ZIC-approximate-capacity}
We begin by first treating the channel as being fixed and time-invariant. The MIMI region is the same as the capacity region for the fixed channel and we provide inner and outer bounds to it by exploiting our more general results on inner and outer bounds on the capacity region that are within a bounded gap for the time-invariant MIMO IC in \cite{Sanjay_Varanasi_constant_gap_to_capacity}. The simple idea is that the MIMO Z-IC is a special case of the MIMO IC corresponding to the channel matrix from the first transmitter to the second receiver being identically equal to zero. Since the bounded gap in the result of \cite{Sanjay_Varanasi_constant_gap_to_capacity} is independent of channel parameters it also holds for the MIMO Z-IC. The following lemmas specify the inner and outer bounds for completeness and easy reference in the GDMT analysis.
\begin{lemma}
\label{lem:ZIC-capacity-upper-bound}
The capacity (or MIMI) region of the $(M_1,N_1,M_2,N_2)$ Z-IC with CSIT, for a given realization of channel matrices $\mathcal{H}$, denoted by $\mathcal{C}(\mathcal{H},\bar{\rho})$, is contained in or outer bounded by the set of real-tuples $\mathcal{R}^u(\mathcal{H},\bar{\rho})$ defined as the set of non-negative rate pairs $(R_1,R_2)$ that satisfy the following three constraints:
\begin{IEEEeqnarray*}{l}
\label{eq_bound1}
R_i\leq \log \left| \left(I_{N_i}+\rho_{ii} H_{ii}H_{ii}^{\dagger}\right)\right|\triangleq I_{bi}, ~i\in\{1,2\};\\
R_1+R_2\leq \log \left| \left(I_{N_1}+\rho_{21} H_{21}H_{21}^{\dagger}+\rho_{11}  H_{11}H_{11}^{\dagger}\right)\right|+ \log \left| \left(I_{N_2}+\rho_{22}  H_{22} \left(I_{M_2}+\rho_{21}  H_{21}^{\dagger}H_{21}\right)^{-1}H_{22}^{\dagger}\right)\right|\triangleq I_{bs}. \label{eq_bound3}
\end{IEEEeqnarray*}
\end{lemma}

The above result is obtained simply by substituting $H_{12}=0_{N_2\times M_1}$ in Lemma 1 of 
\cite{Sanjay_Varanasi_constant_gap_to_capacity}. In what follows, we find an inner bound to the MIMI region which is the achievable rate region of a specific coding scheme where the first transmitter uses a random Gaussian code book and the second transmitter uses a two-level linear Gaussian superposition code as follows
\begin{equation}
\begin{array}{c}
X_2=U_2+W_2,
\end{array}
\end{equation}
where $U_2$ (hereafter called the private sub-message) and $W_2$ (the public sub-message) are mutually independent complex Gaussian random vectors with covariance matrices given as
\begin{IEEEeqnarray}{c}
\label{eq_power_split}
\mathbb{E}(X_1X_1^{\dagger})=I_{M_1}, \quad \mathbb{E}(W_2W_2^{\dagger})=\frac{I_{M_2}}{2} \quad \textrm{and} \quad 
\mathbb{E}(U_2U_2^{\dagger})=\frac{1 }{2}\left(I_{M_2}+\rho^{\alpha_{21}} H_{21}^{\dagger}H_{21}\right)^{-1}.
\end{IEEEeqnarray}

Note that this covariance split satisfies the power constraint in equation \eqref{power_constraint}. The above coding scheme is clearly a specific Han-Kobayashi coding scheme where further the first transmitter's message is treated entirely as a public message. While the GDMT characterized in this paper is the fundamental GDMT of the MIMO Z-IC when both transmitters have CSIT, it will be clear soon that the above coding scheme which uses the knowledge of CSI only to the extent that $Tx_2$ needs the knowledge of $H_{21}$, can achieve the full-CSIT GDMT.

\begin{lemma}
\label{lem:ZIC-achievable-region}
For a given channel realization $\mathcal{H}$, the above described coding scheme can achieve a region $\mathcal{R}^l\left(\mathcal{H},\bar{\rho}\right)$ of non-negative rate pairs $(R_1,R_2)$ that satisfies the following constraints:
\begin{IEEEeqnarray*}{rl}
R_i &\leq \log \left| \left(I_{N_i}+\rho_{ii} H_{ii}H_{ii}^{\dagger}\right)\right|-n_i\triangleq I_{li}, ~i\in\{1,2\};\\
R_1+R_2 &\leq \log \left| \left(I_{N_1}+\rho_{21} H_{21}H_{21}^{\dagger}+\rho_{11}  H_{11}H_{11}^{\dagger}\right)\right|+ \log \left| \left(I_{N_2}+\rho_{22}  H_{22} \left(I_{M_2} + \rho_{21}  H_{21}^{\dagger}H_{21}\right)^{-1}H_{22}^{\dagger}\right)\right| - (n_1+n_2)\triangleq I_{ls},
\end{IEEEeqnarray*}
where
\begin{equation}
\label{eq:def-ni}
n_i=\max\left\{\left(m_{ii}\log(M_i)+m_{ij}\log(M_i+1)\right),\min\{N_i,M_s\}\log(M_x)\right\}+\hat{m}_{ji}, ~\textrm{for}~ 1\leq i\neq j\leq 2
\end{equation}
with $M_x=\max\{M_1,M_2\}$, $M_s=(M_1+M_2)$, $m_{ij}$ representing the rank of the matrix $H_{ij}$, and $\hat{m}_{ij}=m_{ij}\log\left(\frac{(M_i+1)}{M_i}\right)$. Note that $m_{ij}\leq \min\{M_i,N_j\}$.
\end{lemma}
The above result follows by substituting $H_{12}=0_{N_2\times M_1}$ in Lemma~4 of \cite{Sanjay_Varanasi_constant_gap_to_capacity}.
The inner and outer bounds to the MIMI region of the quasi-static fading MIMO Z-IC of Lemmas \ref{lem:ZIC-capacity-upper-bound} and \ref{lem:ZIC-achievable-region}, which are within a bounded gap from each other, can be used to derive its GDMT as we show next. With the outage event $\mathcal{O}$ defined as the set of channels for which the MIMI region does not contain a given rate pair, i.e., $  \mathcal{O} \triangleq \{\mathcal{H}:(R_1,R_2)\notin \mathcal{C}(\mathcal{H},\bar{\rho})\} $, and following a method similar to that of \cite{tse1}, it can be easily proved that $  \mathcal{P}^*_e\left(\bar{\rho},r_1,r_2\right)\dot{=}\Pr\left(\mathcal{O}\right) $, where $\mathcal{P}^*_e\left(\bar{\rho},r_1,r_2\right)$ is the best case average probability of error achievable on the MIMO Z-IC as defined in Section \ref{subsec:def-dmt}. Now, from Lemma~\ref{lem:ZIC-capacity-upper-bound} and \ref{lem:ZIC-achievable-region} we have that for any realization of the channel matrices $\mathcal{H}$, $ \mathcal{R}^l(\mathcal{H},\bar{\rho})\subseteq   \mathcal{C}(\mathcal{H},\bar{\rho})\subseteq   \mathcal{R}^u(\mathcal{H},\bar{\rho}) $ from which we have the sequence of implications below: 
\begin{IEEEeqnarray}{rl}\label{xxx}
   &\{\mathcal{H}:(R_1,R_2)\notin \mathcal{R}^u(\mathcal{H},\bar{\rho})\} \subseteq  \mathcal{O} \subseteq \{\mathcal{H}:(R_1,R_2)\notin \mathcal{R}^l(\mathcal{H},\bar{\rho})\};\nonumber\\
  \quad  \longrightarrow ~ &\Pr\left\{(R_1,R_2)\notin \mathcal{R}^u(\mathcal{H},\bar{\rho})\right\} \dot{\leq}  \mathcal{P}^*_e\left(\bar{\rho},r_1,r_2\right) \dot{\leq} \Pr\left\{(R_1,R_2)\notin \mathcal{R}^l(\mathcal{H},\bar{\rho})\right\};\nonumber\\
   \quad \longrightarrow ~&\Pr\left\{(R_1,R_2)\notin \mathcal{R}^u(\mathcal{H},\bar{\rho})\right\} \dot{\leq}  \rho^{-d_{\textrm{Z-IC}}^*(r_1,r_2)} \dot{\leq} \Pr\left\{(R_1,R_2)\notin \mathcal{R}^l(\mathcal{H},\bar{\rho})\right\};\nonumber \\
   \quad \longrightarrow ~&\Pr\left\{\cup_{i}\{ I_{bi}\leq R_i\}\right\} \dot{\leq}  \rho^{-d_{\textrm{Z-IC}}^*(r_1,r_2)} \dot{\leq} \Pr\left\{\cup_{i} \{ I_{li}\leq R_i\}\right\};\nonumber\\
\quad \longrightarrow ~&\max_{i\in\{1,2,s\}}\Pr\left\{I_{bi}\leq R_i\right\} \dot{\leq}  \rho^{-d_{\textrm{Z-IC}}^*(r_1,r_2)} \dot{\leq} \sum_{i=1,2,s}\Pr\left\{I_{li}\leq R_i\right\},\\
\label{eq:bounds-on-outage-probability}
\quad \longrightarrow ~&\max_{i\in\{1,2,s\}}\Pr\left\{I_{bi}\leq R_i\right\} \dot{\leq}  \rho^{-d_{\textrm{Z-IC}}^*(r_1,r_2)} \dot{\leq} \max_{i\in\{1,2,s\}}\Pr\left\{I_{li}\leq R_i\right\},
\end{IEEEeqnarray}
where $R_{s}=(R_1+R_2)$ and $I_{bi}$'s and $I_{li}$'s are as defined in Lemma~\ref{lem:ZIC-capacity-upper-bound} and Lemma~\ref{lem:ZIC-achievable-region}, respectively. Note that, $I_{bi}=I_{li}+n_i$, for $i=1,2$ and $I_{bs}=I_{ls}+(n_1+n_2)$ where $n_i$'s given by equation \eqref{eq:def-ni} are constants independent of $\rho$, for $i\in \{1,2\}$, which becomes insignificant at asymptotic SNR. Therefore, at asymptotic values of $\rho$ equation \eqref{eq:bounds-on-outage-probability} is equivalent to
\begin{equation*}
 \rho^{-d_{\textrm{Z-IC}}^*(r_1,r_2)}\dot{=} \max_{i\in\{1,2,s\}}\Pr\left\{I_{bi}\leq R_i\right\},
\end{equation*}
which can be written as
\begin{IEEEeqnarray}{rl}
\label{eq_dmt_exp1}
d_{\textrm{Z-IC}}^*(r_1,r_2)=& \min_{i\in \mathcal{I}}{d_{O_i}({r}_i)},
\end{IEEEeqnarray}
where 
\begin{IEEEeqnarray}{rl}
\label{eq:outage_probability}
d_{O_i}({r}_i)=& \lim_{\rho \to \infty} -\frac{\Pr\left(I_{bi}\leq r_i\log(\rho)\right)}{\log(\rho)},
\end{IEEEeqnarray}
for all $i\in \mathcal{I}=\{1,2,s\}$ and $r_s=(r_1+r_2)$.

The only remaining step to characterize the GDMT completely is to evaluate the probabilities in equation \eqref{eq:outage_probability}, which in turn requires the statistics of the mutual information terms $I_{b_i}$'s. It will be shown in the next section that these statistics and therefrom the GDMT of the channel can be characterized, if only the joint distribution of the eigenvalues of two mutually correlated random Wishart matrices are known.

\section{Explicit GDMT of the MIMO Z-IC}
\label{sec_DMT_Z_channel}
In this section, we evaluate the diversity orders, $d_{O_i}(r_i)$'s, of the three outage events given in equation \eqref{eq_dmt_exp1}, which would yield the explicit GDMT of the MIMO Z-IC.
The diversity orders $d_{O_i}(r_i)$ for $i = 1,2 $ are single-user bounds and can be easily obtained from \cite{tse1} and consequently, we have
\begin{IEEEeqnarray}{rl}
d_{O_i}({r}_i)=& \lim_{\rho \to \infty} -\frac{\Pr\left(\sum_{k=1}^{\min\{M_i,N_i\}}(\alpha_{ii}-\upsilon_{i,k})^+\leq r_i\right)}{\log(\rho)},~i\in \{1,2\},
\end{IEEEeqnarray}
where $\upsilon_{i,k}$'s are the negative SNR exponents of the ordered eigenvalues of the matrices $H_{ii}H_{ii}^\dagger$. The joint distribution of $\{\upsilon_{i,k}\}_{k=1}^{\min\{M_i,N_i\}}$ was specified in \cite{tse1}. Using this distribution and a similar technique as in \cite{tse1}, it can be shown that
\begin{subequations}
\label{eq:modified-dmt}
\begin{align}
d_{O_i}({r}_i)= \min & \sum_{k=1}^{\min\{M_i,N_i\}}(M_i+N_i+1-2k)\upsilon_{i,k}\\
\textrm{subject to: } & \sum_{k=1}^{\min\{M_i,N_i\}}(\alpha_{ii}-\upsilon_{i,k})^+\leq r;\\
& 0\leq \upsilon_{i,1}\leq \cdots \leq \upsilon_{i,\min\{M_i,N_i\}}.
\end{align}
\end{subequations}
Since a similar optimization problem will recur in the rest of the paper we formally state the solution of the above problem in the following lemma for convenient future reference.
\begin{lemma}
\label{lem:m-dmt}
If $d(r)$ represents the solution of the optimization problem
\begin{subequations}
\label{eq:m-dmt-optimization-problem}
\begin{align}
\label{eq:m-dmt-optimization-problem-a}
& \min \sum_{i=1}^{m}(M+N+1-2i)\mu_i\\
\textrm{subject to: } & \sum_{i=1}^{m}(\alpha-\mu_i)^+\leq r;\\
& 0\leq \mu_1\leq \cdots \leq \mu_m,
\end{align}
\end{subequations}
then,
\begin{equation}\label{eq:m-dmt-optimization-solution}
    d(r)=\alpha \cdot d_{M,N}\left(\frac{r}{\alpha}\right),~\textrm{for}~0\leq r\leq m\alpha,
\end{equation}
where $m=\min\{M,N\}$ and $d_{M,N}(r)$ represents the DMT of a $M\times N$ point-to-point channel and is a piecewise linear curve joining the points $(M-k)(N-k)$ for $k=0,1,\cdots m$.
\end{lemma}

\begin{proof}[Proof]
Putting $\mu_i^{'}=\frac{\mu_i}{\alpha}$ in the optimization problem \eqref{eq:m-dmt-optimization-problem} we get,
\begin{subequations}
\begin{align}
\frac{d(r)}{\alpha}=& \min \sum_{i=1}^{m}(M+N+1-2i)\mu_i^{'}\\
\textrm{subject to: } & \sum_{i=1}^{m}(1-\mu_i^{'})^+\leq \frac{r}{\alpha};\\
& 0\leq \mu_1^{'}\leq \cdots \leq \mu_m^{'}.
\end{align}
\end{subequations}
The solution of this modified optimization problem was derived in \cite{tse1} and is given by
\begin{IEEEeqnarray*}{rl}
\frac{d(r)}{\alpha}= &d_{M,N}\left(\frac{r}{\alpha}\right),~\textrm{for}~0\leq \frac{r}{\alpha}\leq m,\\
\textrm{or,}~    d(r)= &\alpha \, d_{M,N}\left(\frac{r}{\alpha}\right),~\textrm{for}~0\leq r\leq m\alpha.
\end{IEEEeqnarray*}
\end{proof}

The solution of the optimization problem \eqref{eq:modified-dmt} is now evident from Lemma~\ref{lem:m-dmt}, and is given by
\begin{IEEEeqnarray}{l}
\label{eq:single-user-outage-exponent}
d_{O_i}(r_i)=\alpha_{ii}d_{M_i,N_i}\left(\frac{r_i}{\alpha_{ii}}\right),
~\forall ~ r_i\in [0, \min\{M_i, N_i\}\alpha_{ii}]~\textrm{and}~i\in \{1,2\},
\end{IEEEeqnarray}
where $d_{m,n}(r)$ is the optimal diversity order of a point-to-point MIMO channel with $m$ transmit and $n$ receive antennas, at integer values of $r$ and is point wise linear between integer values of $r$. To evaluate $d_{O_s}(r_s)$, we write  the bound $I_{bs}$ of Lemma~\ref{lem:ZIC-capacity-upper-bound} in the following way
\begin{IEEEeqnarray*}{ll}
I_{bs} = \log \left| \left(I_{M_1}+\rho_{11}H_{11}^{\dagger}\left(I_{N_1}+\rho_{21} H_{21}H_{21}^{\dagger}\right)^{-1}H_{11}\right)\right| +\log \left| \left(I_{N_2}+\rho_{22}  H_{22} \left(I_{M_2}+\rho_{21}  H_{21}^{\dagger}H_{21}\right)^{-1}H_{22}^{\dagger}\right)\right|\\
\qquad \qquad \qquad \qquad \qquad \qquad + \log \left|\left(I_{N_1}+\rho_{21} H_{21}H_{21}^{\dagger}\right)\right|,\\
\stackrel{(a)}{=}  \left\{\sum_{i=1}^{p}(1+\rho^{\alpha_{21}} \lambda_i)+\sum_{j=1}^{q_1}(1+\rho^{\alpha_{11}} \mu_j)+\sum_{k=1}^{q_2}(1+\rho^{\alpha_{22}} \pi_k)\right\},
\end{IEEEeqnarray*}
where in step $(a)$, we define $p=\min\{M_2, N_1\}$, $q_1=\min\{M_1, N_1\}$, $q_2=\min\{M_2, N_2\}$ and denote the ordered non-zero (with probability $1$) eigenvalues of the three matrices $W_1=H_{11}^{\dagger}\left(I_{N_1}+\rho_{21} H_{21}H_{21}^{\dagger}\right)^{-1}H_{11}$, $W_2=H_{22} \left(I_{M_2}+\rho_{21}  H_{21}^{\dagger}H_{21}\right)^{-1}H_{22}^{\dagger}$ and $W_3=H_{21}H_{21}^{\dagger}$ by $\mu_1\geq \cdots \geq \mu_{q_1}>0$, $\pi_1\geq \cdots \geq \pi_{q_2} >0$ and $\lambda_1\geq \cdots \geq \lambda_{p} >0$, respectively. Now, using the transformations $\lambda_{i}=\rho^{-\upsilon_i}$, for $1\leq i\leq p$, $\mu_{j}=\rho^{-\beta_j}$, for $1\leq j\leq q_1$ and $\pi_{k}=\rho^{-\gamma_k}, ~1\leq k\leq q_2$ in the above equation and substituting that in turn in equation \eqref{eq:outage_probability} we get
\begin{IEEEeqnarray}{rl}
\label{eq:outage-sum-bound-FCSIT}
\rho^{-d_{O_s}(r_s)}\dot{=}&\Pr\left( \left\{\sum_{i=1}^{p}(\alpha_{21}- \upsilon_i)^++\sum_{j=1}^{q_1}(\alpha_{11}-\beta_j)^+
 +\sum_{k=1}^{q_2}(\alpha_{22}-\gamma_k)^+\right\}< r_s\right).
\end{IEEEeqnarray}
To evaluate this expression we need to derive the joint distribution of $\vec{\gamma},\vec{\beta}$ and $\vec{\upsilon}$ where
$\vec{\gamma}=\{\gamma_1,\cdots ,\gamma_{q_2}\}$ and similarly $\vec{\upsilon}=\{\upsilon_1,\cdots ,\upsilon_p\}$ and $\vec{\beta}=\{\beta_1,\cdots ,\beta_{q_1}\}$. Note that $W_1, W_2$ and $W_3$ are not independent and hence neither are $\vec{\gamma},\vec{\beta}$ and $\vec{\upsilon}$. However, this distribution can be computed using Theorems 1 and 2 of \cite{DDF-DMT-Journal-IT}. Using this joint distribution, equation \eqref{eq:outage-sum-bound-FCSIT} and a similar argument as in \cite{tse1} $d_{O_s}(r_s)$ can be derived as the solution of a convex optimization problem as stated in the following lemma.
\begin{lemma}
\label{lem:ZIC-optimization-problem}
The negative SNR exponent of the outage event corresponding to the sum bound in Lemma~\ref{lem:ZIC-achievable-region}, i.e., $d_{O_s}(r_s)$, is equal to the minimum of the following objective function:
\begin{IEEEeqnarray}{rl}
\mathcal{F}_{(M_1,N_1,M_2,N_2)}^{\textrm{FCSIT}}= & \sum_{i=1}^{p}(M_2+N_1+M_1+N_2+1-2i)\upsilon_i +\sum_{j=1}^{q_1}(M_1+N_1+1-2j)\beta_j\nonumber \\ &+\sum_{k=1}^{q_2}(M_2+N_2+1-2k)\gamma_k
  -(M_1+N_2)p\alpha_{21}\nonumber\\
\label{eq:main-optimization-a}
&+\sum_{k=1}^{q_2}\sum_{i=1}^{\min\{(M_2-k), N_2\}}(\alpha_{21}-\upsilon_i-\gamma_k)^+
+\sum_{j=1}^{q_1}\sum_{i=1}^{\min\{(N_1-j), M_1\}}(\alpha_{21}-\upsilon_i-\beta_j)^+;
\end{IEEEeqnarray}
constrained by
\begin{subequations}
\label{eq:main-optimization}
\begin{align}
& \sum_{i=1}^{p}(\alpha_{21}- \upsilon_i)^++\sum_{j=1}^{q_1}(\alpha_{11}-\beta_j)^++\sum_{k=1}^{q_2}(\alpha_{22}-\gamma_k)^+< r_s;\\
& 0\leq \upsilon_1\leq \cdots \leq \upsilon_{p};\\
& 0\leq \beta_1\leq \cdots \leq \beta_{q_1};\\
& 0\leq \gamma_1\leq \cdots \leq \gamma_{q_2};\\
& (\upsilon_i+\beta_j)\geq \alpha_{21}, ~\forall (i+j)\geq (N_1+1);\\
& (\upsilon_i+\gamma_k)\geq \alpha_{21}, ~\forall (i+k)\geq (M_2+1).
\end{align}
\end{subequations}
\end{lemma}
\begin{proof}
The proof is relegated to Appendix~\ref{pf:lem:ZIC-optimization-problem}.
\end{proof}
Upon differentiating the objective function with respect to $\{\upsilon_i\}$, $\{\beta_j\}$ and $\{\gamma_k\}$, it can be easily verified that \eqref{eq:main-optimization-a} is a convex function of these variables. The constraints on the other hand are all linear. Therefore, equations \eqref{eq:main-optimization-a} and \eqref{eq:main-optimization} represent an convex optimization problem (e.g., see Section $4.2.1$ in \cite{BV}) and hence can be solved efficiently using numerical methods. Since we have already found expressions for $d_{\mathcal{O}_i}$ for $i=1,2$ as in equation~\eqref{eq:single-user-outage-exponent}, Lemma~\ref{lem:ZIC-optimization-problem} provides the last piece of the puzzle required to characterize the GDMT of the MIMO Z-IC, by evaluating $d_{\mathcal{O}_s}$. This is stated formally in the following theorem.

\begin{thm}
\label{thm:DMT-ZIC-FCSIT}
The GDMT of the $(M_1,N_1,M_2,N_2)$ Z-IC under the CSIT assumption with a short term power power constraint is given as
 \begin{equation*}
    d_{\textrm{Z-IC}}^{*,\textrm{CSIT}}(r_1,r_2)=\min_{i\in \{1,2,s\}} ~d_{\mathcal{O}_i}(r_i),
 \end{equation*}
 where $d_{\mathcal{O}_i}(r_i)$ for $i=1,2$ and $i=s$ are given by equation \eqref{eq:single-user-outage-exponent} and Lemma~\ref{lem:ZIC-optimization-problem}, respectively.
\end{thm}

Although the computation of $d_{\mathcal{O}_s}(r_s)$ and hence the characterization of the GDMT of a general Z-IC with an arbitrary number of antennas at each node requires application of numerical methods, in what follows, we shall provide closed form expressions for it for several special cases. Since the GDMT with CSIT acts as an upper bound to the GDMT under the No-CSIT assumption (i.e., with only CSIR), these expressions facilitate an easy characterization of the gap between perfect and No-CSIT performances of the channel. The No-CSIT GDMT of the channel will be characterized in Section~\ref{sec:ZIC-no-CSIT} for different range of values of $\alpha_{21}$ and number of antennas at different nodes.

The central idea in all of the following proofs is the fact that the steepest descent method provides a global optimal value of a convex optimization problem: this is so because the steepest descent method provides a local optimal value of the objective function and a local optimum is equal to a global optimum for a convex function~\cite{BV}. 
The first case considered is a class of channels where all the nodes have equal number of antennas.

\begin{thm}[The symmetric GDMT of the $(n,n,n,n)$ Z-IC]
\label{thm:DMT-symmetric-Z-FCSIT}
Consider the MIMO Z-IC with $M_1=M_2=N_1= N_2=n$ and SNRs and INRs of different links are as described in Section \ref{sec_channel_model} with $\alpha_{21}=\alpha$ and $\alpha_{22}=1=\alpha_{11}$. The best achievable diversity order on this channel, with CSIT and short term average power constraint \eqref{power_constraint}, at multiplexing gain pair $(r_1,r_2)$, is given by
\begin{equation}
d_{\textrm{Z-IC},(n,n,n,n)}^{*,\textrm{CSIT}}(r_1,r_2)=\min \left\{d_{n,n}(r_1),d_{n,n}(r_2),d_{s,(n,n,n,n)}^{\textrm{CSIT}}(r_s)\right\},
\end{equation}
where $d_{\mathcal{O}_s}(r_s)$ for the special channel configuration being considered is denoted by $d_{s,(n,n,n,n)}^{\textrm{CSIT}}(r_s)$, and is given as
\begin{IEEEeqnarray}{l}
\label{eq:dmt-symmetric-a-leq-1}
d_{s,(n,n,n,n)}^{\textrm{CSIT}}(r_s)=\left\{\begin{array}{l}
\alpha d_{n,3n}(\frac{r_s}{\alpha})+2n^2(1-\alpha), ~\textrm{for}~ 0\leq r_s \leq n\alpha;\\
2(1-\alpha) d_{n,n}(\frac{(r_s-n\alpha)}{2(1-\alpha)}), ~\textrm{for}~ n\alpha\leq r_s \leq n(2-\alpha).
\end{array}\right.
\end{IEEEeqnarray}
If $\alpha\leq 1$ and
\begin{IEEEeqnarray}{l}
\label{eq:dmt-symmetric-a-geq-1}
d_{s,(n,n,n,n)}^{\textrm{CSIT}}(r_s)=\left\{\begin{array}{l}
d_{n,3n}(r_s)+n^2(\alpha-1),~0\leq r_s\leq n; \\
(\alpha-1)d_{n,n}\left(\frac{r_s-n}{\alpha-1}\right),~n\leq r_s\leq n\alpha,
\end{array}\right.
\end{IEEEeqnarray}
for $1\leq \alpha$.
\end{thm}

\begin{proof}[Proof of Theorem~\ref{thm:DMT-symmetric-Z-FCSIT}]
The proof is relegated to Appendix~\ref{pf:thm:DMT-symmetric-Z-FCSIT}.
\end{proof}

\begin{rem}
Note that for $\alpha=1$, the DMT of the Z-IC becomes
\[
 d_{s,(n,n,n,n)}^{\textrm{CSIT}}(r_s)=\min \left\{d_{n,n}(r_1),d_{n,n}(r_2),d_{n,3n}(r_1+r_2)\right\}
 \]
which is exactly the upper bound to the DMT of the 2-user MIMO IC with $n$ antennas at each node, derived in~\cite{EaOlCv}. Since in the Z-IC, the second receiver is free of interference, the DMT of the Z-IC serves as an upper bound to the DMT of the two-user IC.
\end{rem}

\subsection{The DMT of a Femto Cell}
\label{subsec:ZIC-analytic-solutions}
A practical communication channel following the Z-IC network model appears in the so called femto cell environment. The femto cell~\cite{Femto} concept is an outcome of the telecommunication industry's efforts to provide high-throughput, high quality services into the user's home. Consider the scenario depicted in Fig. \ref{figure:femtocell}, where the larger circle represents the macro cell serviced by the macro cell base station (MCBS). Within this macro cell is the smaller circle that represents a small area where the signal from the MCBS either does not reach with enough strength or does not reach at all, hereafter referred to as the femto cell. To provide coverage in this region a smaller user-deployed base station connected to the backbone can be used, which is called the femto cell BS (FCBS). This FCBS can provide mobile services to the users within the femto cell just like a Wi-Fi access point. The basic difference between the FCBS and Wi-Fi access point is that the former operates in a licensed band.

\begin{figure}[!hbt]
\centering
\includegraphics[width=12.0cm,height=8cm,]{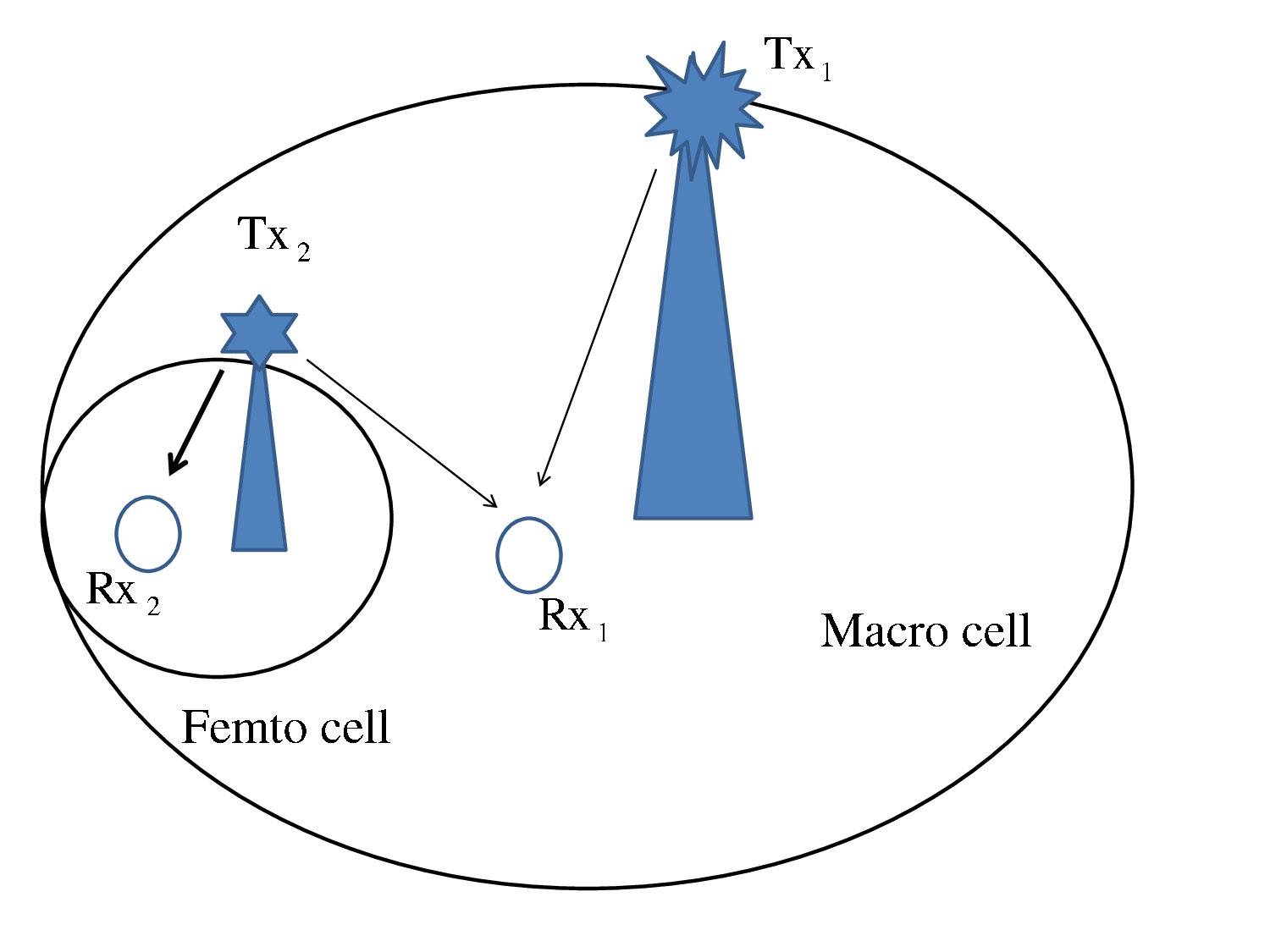}
\caption{Femto cell channel model: downlink.}
\label{figure:femtocell}
\end{figure}

Now, let us consider the downlink communication on such a channel with one mobile user in both the femto cell and the macro-cell. Note that since the MCBS signal does not reach the mobile user within the femto cell, the signal input-output follows the Z-IC model. To model the larger SNR of the femto cell direct link we can assume that $\alpha_{22}\geq 1$. In what follows, we shall derive the DMT of this channel.

\begin{thm}[Femto cell GDMT]
\label{thm:DMT-femto-symmetric}
Consider a Z-IC with $M_1=M_2=N_1= N_2=n$, $\alpha_{22}=\alpha\geq 1$ and $\alpha_{11}=\alpha_{21}=1$, under the CSIT assumption and a short term average power constraint \eqref{power_constraint}. The optimal diversity order of this channel at a multiplexing gain pair $(r_1,r_2)$ is given by
\begin{equation*}
d_{(n,n,n,n)}^{\textrm{Femto}}(r_1,r_2)=\min \left\{d_{n,n}(r_1),\alpha d_{n,n}\left(\frac{r_2}{\alpha}\right),d_{s,(n,n,n,n)}^{\textrm{Femto}}(r_s)\right\}
\end{equation*}
where $d_{\mathcal{O}_s}(r_s)$ for the special channel configuration being considered, is denoted by $d_{s,(n,n,n,n)}^{\textrm{Femto}}(r_s)$, and
\begin{IEEEeqnarray*}{l}
d_{s,(n,n,n,n)}^{\textrm{Femto}}(r_s)=\left\{\begin{array}{l}
d_{n,3n}(r_s)+n^2(\alpha-1), ~\textrm{for}~ 0\leq r_s \leq n;\\
(\alpha-1) d_{n,n}(\frac{(r_s-n)}{(\alpha-1)}), ~\textrm{for}~ n\leq r_s \leq n\alpha.
\end{array}\right.
\end{IEEEeqnarray*}
\end{thm}
\begin{proof}
The proof is relegated to Appendix~\ref{pf:thm:DMT-femto-symmetric}.
\end{proof}

\begin{rem}
Note that the fundamental DMT of the Z-IC with single antenna nodes and $\alpha_{22}=\alpha, \alpha_{11}=\alpha_{21}=1$ was derived in~\cite{SJJ}. This clearly is a special case of Theorem~\ref{thm:DMT-femto-symmetric} and can be obtained by putting $n=1$.
\end{rem}

Typically, in a multiuser communication scenario one end -- say the base station in a cellular network -- can host more antennas than the other. Motivated by this fact in what follows we consider a case which addresses the DMT of such a practical communication network, i.e., where $M_1=M_2=M\leq \min\{ N_1, N_2\}$.

\begin{thm}
\label{thm:asymmetric-FCSIT}
Consider the Z-IC with $M_1=M_2=M \leq \min \{N_1, N_2\}$, $\alpha_{11}=\alpha_{22}=\alpha_{21}= 1$, and with CSIT and the short term average power constraint of \eqref{power_constraint}. The optimal diversity order achievable on this channel at a multiplexing gain pair $(r_1,r_2)$ is given by
\begin{equation*}
d_{M,N_1,M,N_2}^{\textrm{CSIT}}(r_1,r_2)=\min \left\{d_{M,N_1}(r_1),d_{M,N_2}(r_2),d_{s,(M,N_1,M,N_2)}^{\textrm{FCSIT}}(r_s)\right\},
\end{equation*}
where $d_{\mathcal{O}_s}(r_s)$ for the special channel configuration being considered, is denoted by $d_{s,(M,N_1,M,N_2)}^{\textrm{CSIT}}(r_s)$, and is given as
\begin{IEEEeqnarray*}{l}
d_{s,(M,N_1,M,N_2)}^{\textrm{CSIT}}(r_s)=
\left\{\begin{array}{l}
d_{M,(M+N_1+N_2)}({r_s})+M(N_1-M); ~0\leq r_s \leq M;\\
d_{2M,N_1}(r_s); ~ M\leq r_s\leq \min\{N_1,2M\}.
\end{array}\right.
\end{IEEEeqnarray*}
\end{thm}
\begin{proof}
The proof is relegated to Appendix~\ref{pf:thm:asymmetric-FCSIT}.
\end{proof}

\begin{figure}[!hbt]
\centering
\includegraphics[width=12.0cm,height=8cm,]{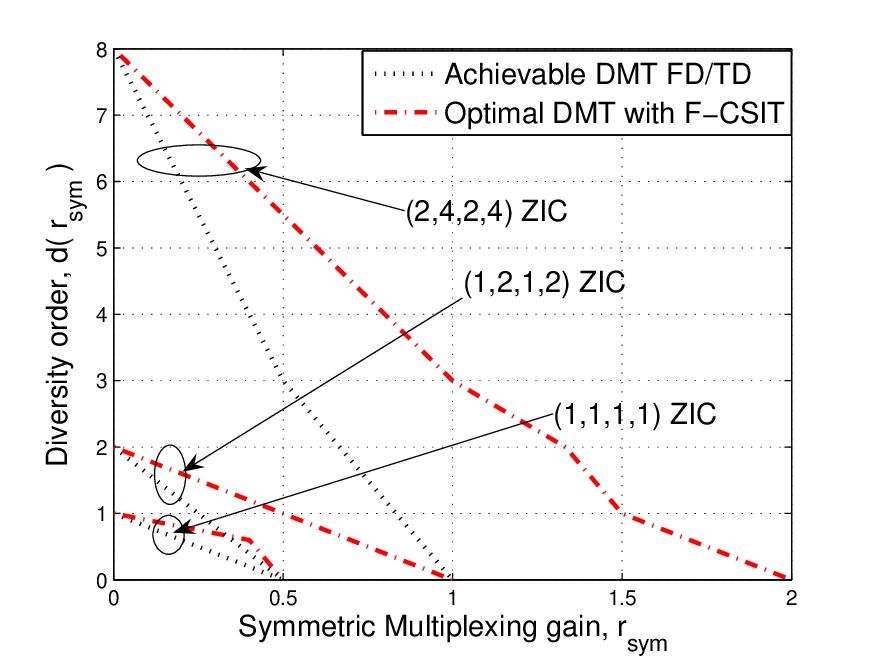}
\caption{The DMT of different Z-ICs with $\bar{\alpha}=[1,1,1]$.}
\label{figure_IC_vs_FDTD}
\end{figure}

Let us now quantify the loss due the use of suboptimal coding schemes on Z-ICs, with respect to the fundamental DMT of the channel achievable by sophisticated coding schemes such as the superposition coding scheme described in Section~\ref{subsec:ZIC-approximate-capacity}. In Fig. \ref{figure_IC_vs_FDTD}, explicit CSIT DMT curves for a few antenna configurations are plotted and compared against the performance of orthogonal schemes such as frequency division (FD) or time division (TD) multiple-access which do not utilize CSIT. It can be noticed from the figure that the gain due to CSIT, over the orthogonal access schemes can be significant, particularly in MIMO Z-ICs. 

While this gap can be reduced by using better coding-decoding schemes, in general with CSIT a better performance can be achieved. However, to evaluate this gap in performance due to lack of CSIT exactly, it is necessary to know the best GDMT achievable on the channel without CSIT. One way to characterize the No-CSIT GDMT would be to either evaluate the MIMI region without CSIT exactly or find inner and outer bounds which are within a constant number of bits to each other, as was done in the CSIT case. For the no CSIT case however, both of these are challenging problems. However, in the following section, we bypass this approach and characterize the No-CSIT GDMT of the MIMO Z-IC for some specific values of the different channel parameters such as the number of antennas at the different nodes and the INR parameter $\alpha$.


\section{GDMT with No CSIT}
\label{sec:ZIC-no-CSIT}
We start by noting that the GDMT under the CSIT assumption derived in the previous sections can serve as an upper bound to the No-CSIT GDMT of the channel. Then we derive the achievable GDMT of a specific No-CSIT transmit-receive scheme, which for two special classes of Z-ICs meets the upper bound, and therefore, represents the fundamental No-CSIT GDMT of the corresponding Z-ICs. 

Consider the following No-CSIT transmit-receive scheme.
Let both the users encode their messages using independent Gaussian signals. Moreover, the decoder at $Rx_1$ does joint maximum-likelihood (ML) decoding of both the messages. However, since $Rx_1$ is not interested in the signal transmitted by $Tx_2$, the event where only the second user's message is decoded incorrectly is not considered as an error event. $Rx_2$ on the other hand uses an ML decoder to decode its own message. Hereafter, we will refer to this scheme as the {\it Individual ML (IML) decoder} and the encoding-decoding scheme as the Independent-coding/IML-decoding scheme, or simply as the IIML scheme.

An achievable rate region of the IIML scheme is given by the following set of rate tuples
\begin{subequations}
\label{eq:IML-achievable-rate-region}
\begin{align}
\mathcal{R}_{\textrm{IIML}}=\Big\{ (R_1,R_2): & R_1\leq \log\det\left(I_{N_1}+\frac{\rho}{M_1}H_{11}H_{11}^\dagger\right)\triangleq I_{c_1};\\
& R_2\leq \log\det\left(I_{N_2}+\frac{\rho}{M_2}H_{22}H_{22}^\dagger\right)\triangleq I_{c_2};\\
& (R_1+R_2)\leq \log\det\left(I_{N_1}+\frac{\rho}{M_1}H_{11}H_{11}^\dagger+\frac{\rho^\alpha }{M_2}H_{21}H_{21}^\dagger\right)\triangleq I_{c_s};\Big\}
\end{align}
\end{subequations}
Note in the above set of equations we do not have a constraint on $R_2$ due to the point-to-point link from $Tx_2$ to $Rx_1$ because of the IML decoding definition, i.e., $Rx_2$ does not consider it as an error event if only the message of $Tx_2$ is decoded erroneously. Using the above expression for the achievable rate region, the corresponding achievable DMT of this IIML transmit-receive scheme can be easily computed using standard techniques. The result is specified in the following lemma.
\begin{lemma}
\label{lem:IML-dmt-exponents}
If we denote the achievable diversity order of the IIML scheme, at multiplexing gain pair $(r_1,r_2)$, by $d_{(M_1,N_1,M_2,N_2)}^{\textrm{IIML}}(r_1,r_2)$ then
\begin{equation}
\label{eq:IML-dmt-exponents}
    d_{(M_1,N_1,M_2,N_2)}^{\textrm{IIML}}(r_1,r_2)~\geq~\min_{i\in\{1,2,s\}}\{d_{i,(M_1,N_1,M_2,N_2)}^{\textrm{IIML}}(r_i) \},
\end{equation}
where $r_s=(r_1+r_2)$ and
\begin{equation}
\label{eq:def-outage-exponent-IML-decoder}
    d_{i,(M_1,N_1,M_2,N_2)}^{\textrm{IIML}}(r_i)=\lim_{\rho\to\infty} -\frac{\log\left(\Pr\left(I_{ci}\leq r_i\right)\right)}{\log(\rho)},~\forall~i\in \{1,2,s\}.
\end{equation}
\end{lemma}
\begin{proof}
The proof is relegated to Appendix~\ref{pf:lem:IML-dmt-exponents}.
\end{proof}
Note that $I_{c_1}$ and $I_{c_2}$ represent the mutual information of a point-to-point channel with channel matrices $H_{11}$ and $H_{22}$, respectively. Therefore, by the results of \cite{tse1} we have
\begin{IEEEeqnarray}{l}
\label{eq:IML-single-user-exponents}
d_{i,(M_1,N_1,M_2,N_2)}^{\textrm{IIML}}(r_i)=d_{M_i,N_i}(r_i),~0\leq r_i\leq \min\{M_i,N_i\},
\end{IEEEeqnarray}
and $i\in \{1,2\}$. To analyze the probability of the outage event due to the third bound of the achievable rate region we first approximate $I_{c_s}$ by another term which does not differ from it by more than a constant. Note that
\begin{IEEEeqnarray}{rl}
& \log\det\left(I_{N_1}+\rho H_{11}H_{11}^\dagger+\rho^\alpha H_{21}H_{21}^\dagger\right)-N_1\log(\max\{M_1,M_2\}),\nonumber\\
\label{eq:equivalence-of-two-MIs}
& \leq I_{c_s} \leq \log\det\left(I_{N_1}+\rho H_{11}H_{11}^\dagger+\rho^\alpha H_{21}H_{21}^\dagger\right)\triangleq I_{c_s}'.
\end{IEEEeqnarray}
Since a constant independent of $\rho$, does not matter in the high SNR analysis, to compute $d_s^{\textrm{IIML}}$ we can use $I_{c_s}'$ in place of $I_{c_s}$. Next, we write $I_{c_s}'$ in the following manner:
\begin{IEEEeqnarray}{rl}
I_{c_s}'=& \log\det\left(I_{N_1}+\rho H_{11}H_{11}^\dagger+\rho^\alpha H_{21}H_{21}^\dagger\right),\\
=& \log\det\left(I_{M_1}+\rho\widetilde{H}_{11}^\dagger \widetilde{H}_{11}\right)+\log\det\left(I_{N_1}+\rho^\alpha H_{21}H_{21}^\dagger\right),\\
\stackrel{(a)}{=}&\left\{\sum_{j=1}^{q_1}(1+\rho^{\alpha_{11}} \mu_j)+\sum_{i=1}^{p}(1+\rho^{\alpha_{21}} \lambda_i)\right\},
\end{IEEEeqnarray}
where $\widetilde{H}_{11}=\left(I_{N_1}+\rho^\alpha H_{21}H_{21}^\dagger\right)^{-\frac{1}{2}}H_{11}$ and in step $(a)$, $p=\min\{M_2, N_1\}$, $q_1=\min\{M_1,~ N_1\}$. Also, we have denoted the ordered non-zero (with probability $1$) eigenvalues of $W_1=H_{11}^{\dagger}\left(I_{N_1}+\rho_{21} H_{21}H_{21}^{\dagger}\right)^{-1}H_{11}$ and $W_2=H_{21}H_{21}^{\dagger}$ by $\mu_1\geq \cdots \geq \mu_{q_1}>0$ and $\lambda_1\geq \cdots \geq \lambda_{p} >0$, respectively. Now, using the transformations $\lambda_{i}=\rho^{-\upsilon_i}$, for $1\leq i\leq p$, $\mu_{j}=\rho^{-\beta_j}$, for $1\leq j\leq q_1$ in the above equation and substituting that in turn in equation \eqref{eq:def-outage-exponent-IML-decoder} we get
\begin{IEEEeqnarray}{rl}
\rho^{-d_{s,(M_1,N_1,M_2,N_2)}^{\textrm{IIML}}(r_s)}\dot{=}&\Pr\left( \left\{\sum_{i=1}^{p}(\alpha_{21}- \upsilon_i)^++\sum_{j=1}^{q_1}(\alpha_{11}-\beta_j)^+
\right\}< r_s\right).
\end{IEEEeqnarray}
To evaluate this expression we need to derive the joint distribution of $\vec{\beta}$ and $\vec{\upsilon}$. Note that, since $W_1$ and $W_2$ are mutually correlated and so are $\vec{\beta}$ and $\vec{\upsilon}$. As already stated earlier, in general characterizing the joint distribution of the eigenvalues of such mutually correlated random matrices is a hard problem. However, in what follows, we show that this distribution can be computed using Theorems 1 and 2 of \cite{DDF-DMT-Journal-IT}, which in turn facilitates the characterization of $d_s^{\textrm{IIML}}(r_s)$.
\begin{lemma}
\label{lem:IML-sum-optimization}
The diversity order of the probability of the outage event corresponding to the sum bound in \eqref{eq:IML-achievable-rate-region}, i.e., $d_{s,(M_1,N_1,M_2,N_2)}^{\textrm{IIML}}(r_s)$, is equal to the minimum of the following objective function:
\begin{subequations}
\label{eq:IML-sum-optimization}
\begin{align}
\label{eq:IML-sum-optimization-a}
d_{s,(M_1,N_1,M_2,N_2)}^{\textrm{IIML}}(r_s)=\min_{\left(\vec{\upsilon},\vec{\beta} \right)}& \sum_{i=1}^{p}(M_2+N_1+M_1+1-2i)\upsilon_i +\sum_{j=1}^{q_1}(M_1+N_1+1-2j)\beta_j\nonumber \\
  &~~~~~~~~~~~~-M_1p\alpha_{21} +\sum_{j=1}^{q_1}\sum_{i=1}^{\min\{(N_1-j), M_1\}}(\alpha_{21}-\upsilon_i-\beta_j)^+;\\
\label{eq:IML-sum-optimization-b}
\textrm{constrained by: } &
\sum_{i=1}^{p}(\alpha_{21}- \upsilon_i)^++\sum_{j=1}^{q_1}(\alpha_{11}-\beta_j)^+\leq r_s;\\
\label{eq:IML-sum-optimization-c}
& 0\leq \upsilon_1\leq \cdots \leq \upsilon_{p};\\
\label{eq:IML-sum-optimization-d}
& 0\leq \beta_1\leq \cdots \leq \beta_{q_1};\\
\label{eq:IML-sum-optimization-e}
& (\upsilon_i+\beta_j)\geq \alpha_{21}, ~\forall (i+j)\geq (N_1+1).
\end{align}
\end{subequations}
\end{lemma}
\begin{proof}
The proof is relegated to Appendix~\ref{pf:lem:IML-sum-optimization}.
\end{proof}

\begin{thm}[A lower bound to the No-CSIT GDMT of the Z-IC]
\label{thm:DMT-ZIC-noCSIT-lower-bound} 
The diversity order achievable by the IIML scheme on a $(M_1,N_1,M_2,N_2)$ Z-IC without CSIT is given as
 \begin{equation*}
    d_{\textrm{LB, Z-IC}}^{\textrm{No-CSIT}}(r_1,r_2)=\min_{i\in \{1,2,s\}} ~d_{i,(M_1,N_1,M_2,N_2)}^{\textrm{IIML}}(r_i),
 \end{equation*}
where $d_{i,(M_1,N_1,M_2,N_2)}^{\textrm{IIML}}(r_i)$ for $i=1,2$ and $i=s$ are given by equation \eqref{eq:IML-single-user-exponents} and Lemma~\ref{lem:IML-sum-optimization}, respectively. This GDMT also represents a lower bound to the No-CSIT GDMT of the $(M_1,N_1,M_2,N_2)$ Z-IC.
\end{thm}

\begin{proof}
The first part of the theorem follows from Lemma~\ref{lem:IML-dmt-exponents}, and the second part of the theorem follows from the fact that the IIML scheme is only one transmit-receive schemes among others that are possible on the Z-IC.
\end{proof}

Although the computation of $d_{s,(M_1,N_1,M_2,N_2)}^{\textrm{IIML}}(r_s)$ and hence characterization of the lower bound to the No-CSIT GDMT of a general Z-IC with arbitrary number of antennas at each node require the application of numerical methods, in what follows, we shall provide closed-form expressions for it for various special cases. We will see that for two special classes of Z-ICs this lower bound meets the upper bound, i.e., the CSIT GDMT of the channel. We start with the case where all the nodes have equal number of antennas.

\begin{lemma}
\label{lem:DMT-symmetric-Z-achievable}
On a MIMO Z-IC with $n$ antennas at all the nodes, $\alpha_{11}=\alpha_{22}=1$ and $\alpha_{21}=\alpha \geq 1$, the IIML scheme can achieve the following GDMT
\begin{equation*}
d_{(n,n,n,n)}^{\textrm{IIML}}(r_1,r_2)=\min \left\{d_{n,n}(r_1),d_{n,n}(r_2), d_{s,(n,n,n,n)}^{\textrm{IIML}}(r_s)\right\}
\end{equation*}
where $d_{s,(n,n,n,n)}^{\textrm{IIML}}(r_s)$ is given as
\begin{IEEEeqnarray}{l}
\label{eq:dmt-IML-symmetric}
d_{s,(n,n,n,n)}^{\textrm{IIML}}(r_s)=\left\{\begin{array}{l}
d_{n,2n}(r_s)+n^2(\alpha-1), ~0\leq r_s\leq n;\\
(\alpha-1)d_{n,n}(\frac{r_s-n}{\alpha-1}),~n\leq r_s\leq n\alpha.
\end{array}\right.
\end{IEEEeqnarray}
\end{lemma}
\begin{proof}
The desired result is obtained by following the same steps as in the second part of Theorem~\ref{thm:DMT-symmetric-Z-FCSIT}.
\end{proof}

Fig. \ref{figure_Z_DMT_n} illustrates that the IIML scheme can achieve the GDMT (with CSIT) of the MIMO Z-IC on a region of low multiplexing gains.
\begin{figure}[!hbt]
\centering
\includegraphics[width=12.0cm,height=8cm,]{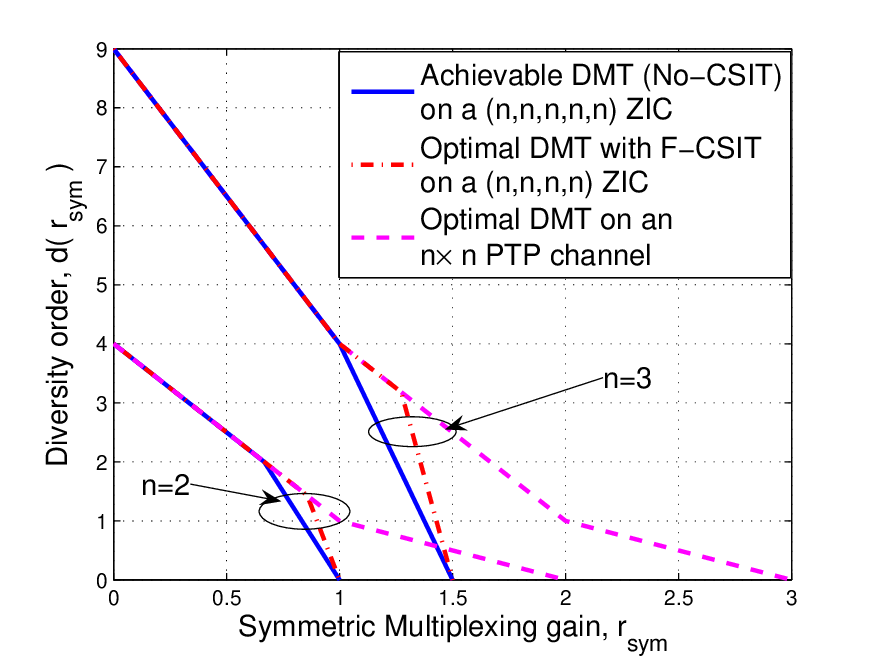}
\caption{Comparison of the DMT on a Z-IC with $\bar{\alpha}=[1,1,1]$ to point-to-point (PTP) performance.}
\label{figure_Z_DMT_n}
\end{figure}

On the other hand, comparing equations \eqref{eq:dmt-symmetric-a-geq-1} to \eqref{eq:dmt-IML-symmetric} we see that the IIML scheme can achieve the CSIT GDMT for high multiplexing gain values - $n\leq r_s\leq n\alpha $ - as well, when $\alpha\geq 1$. This fact raises the natural question: is it possible for the IIML scheme to achieve the CSIT GDMT for all multiplexing gains? If it is, under what circumstances? It turns out that if the interference is strong enough then the IIML decoder can achieve the CSIT GDMT for all symmetric multiplexing gains. For example, Fig. \ref{figure_effect_of_alpha} illustrates this effect on a Z-IC with 2 antennas at all the nodes. These characteristics of the GDMT on MIMO Z-ICs for general $n$ are captured by the following Lemma.

\begin{figure}[!hbt]
\centering
\includegraphics[width=12.0cm,height=8cm,]{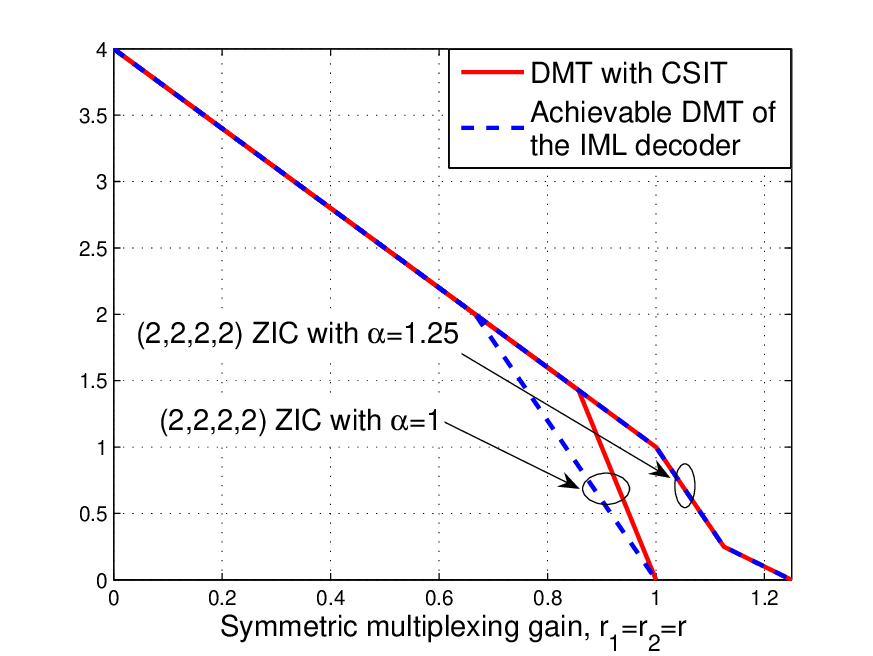}
\caption{The GDMT on the $(2,2,2,2)$ Z-IC with different $\bar{\alpha}$.}
\label{figure_effect_of_alpha}
\end{figure}

\begin{thm}
\label{thm:DMT-symmetric-no-CSIT}
The GDMT specified in Lemma~\ref{lem:DMT-symmetric-Z-achievable} represents the fundamental GDMT of the Z-IC with No-CSIT for the symmetric case of $r_1=r_2=r$, if
\begin{equation}
\label{eq:condition-noCSIT-1}
\alpha \geq 1+\frac{d_{n,n}\left(\frac{n}{2}\right)}{n^2}.
\end{equation}
\end{thm}
\begin{proof}
Detailed proof will be provided in Appendix~\ref{pf:thm:DMT-symmetric-no-CSIT}.
\end{proof}
\vspace{2mm}

\begin{ex}
\label{ex-siso-strong-no-csit-zic}
Specializing the result of the above theorem to the SISO case, i.e., $n=1$, we obtain the symmetric GDMT of the SISO Z-IC with $\bar{\alpha}=[1,~\alpha,~1]$ for $\alpha \geq \frac{3}{2}$, which can be written as
\begin{IEEEeqnarray}{rl}
\label{eq:ex-strong-no-csit-zic}
d_{\textrm{SISO-ZIC}}^*(r_1,r_2)=\min \left\{(1-r_1)^+,(1-r_2)^+,(1-r_s)^++(\alpha-r_s)^+\right\}.
\end{IEEEeqnarray}
This result is identical to the achievable GDMT of the CMO scheme, reported in Section III-B of \cite{Nafea-Seddik-Nafie-Gamal-ZIC}. However, there is a significant difference between the two results. While \cite{Nafea-Seddik-Nafie-Gamal-ZIC} proves only the achievability of the GDMT in equation \eqref{eq:ex-strong-no-csit-zic} on a SISO Z-IC, Theorem~\ref{thm:DMT-symmetric-no-CSIT} above also proves that no better performance in terms of GDMT can be achieved. 
\end{ex}

\begin{lemma}
\label{lem:DMT-IML-asymmetric}
Consider the MIMO Z-IC as in Theorem~\ref{thm:asymmetric-FCSIT}, but with no CSIT. The achievable diversity order of the IIML scheme on this channel, at a multiplexing gain pair $(r_1,r_2)$, is given by
\begin{equation*}
d_{(M,N_1,M,N_2)}^{\textrm{IIML}}(r_1,r_2)=\min \left\{d_{M,N_1}(r_1),d_{M,N_2}(r_2),d_{2M,N_1}(r_s)\right\}.
\end{equation*}
\end{lemma}
\begin{proof}
From Lemma~\ref{lem:IML-dmt-exponents} and equation \eqref{eq:IML-single-user-exponents} it is clear that to prove the lemma it is sufficient to derive an expression for $d_{s,(M,N_1,M,N_2)}^{\textrm{IIML}}(r_s)$. Towards that, for convenience, we use $I_{c_s}^{'}$ instead of $I_{c_s}$ to evaluate the corresponding outage event since the two are within a constant gap which does not matter in the asymptotics of high SNR as was shown in \eqref{eq:equivalence-of-two-MIs}. Using this in \eqref{eq:def-outage-exponent-IML-decoder} along with the facts that $M_1=M_2=M$ and $\alpha=1$ we get
\begin{IEEEeqnarray}{rl}
\rho^{-d_{s,(M,N_1,M,N_2)}^{\textrm{IIML}}(r_s)}\dot{=}&\Pr\left\{\log\det\left(I_{N_1}+\rho H_{11}H_{11}^\dagger +\rho H_{21}H_{21}^\dagger \right)\leq r_s\log(\rho)\right\},\nonumber\\
\label{eq:outage-sum-noCSIT-asymmetric}
=& \Pr \left\{\log\det\left(I_{N_1}+\rho H_eH_e^\dagger \right)\leq r_s\log(\rho)\right\},
\end{IEEEeqnarray}
where $H_e=[H_{11}~H_{21}]\in \mathbb{C}^{N_1\times (2M)}$ is identically distributed as the other channel matrices, since $H_{11}$ and $H_{21}$ are mutually independent. However, the right hand side of the last equation represents the outage probability of an $N_1\times 2M$ point-to-point MIMO channel whose diversity order was computed in \cite{tse1} and is given by $d_{N_1,2M}(r_s)$. Using this in equation \eqref{eq:outage-sum-noCSIT-asymmetric} we get
\begin{IEEEeqnarray*}{rl}
\rho^{-d_{s,(M,N_1,M,N_2)}^{\textrm{IIML}}(r_s)}\dot{=}\rho^{-d_{N_1,2M}(r_s)}\, , \quad {\rm or} \quad
d_{s,(M,N_1,M,N_2)}^{\textrm{IIML}}(r_s)=d_{N_1,2M}(r_s).
\end{IEEEeqnarray*}
Substituting this and equation \eqref{eq:IML-single-user-exponents} into equation \eqref{eq:IML-dmt-exponents} we obtain the desired result.
\end{proof}

\begin{figure}[!hbt]
\centering
\includegraphics[width=12.0cm,height=8cm,]{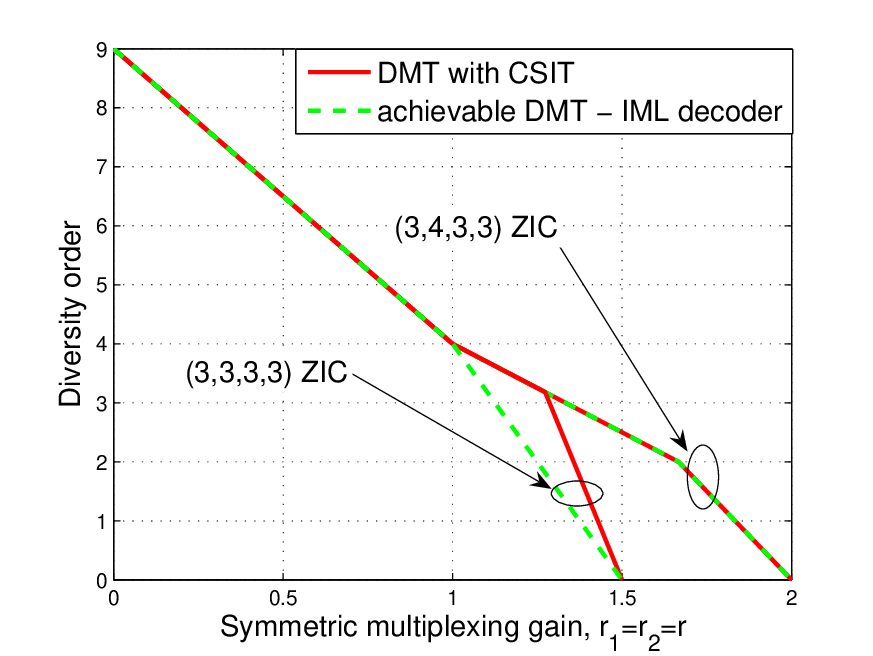}
\caption{Effect of increasing the number of antennas at the interfered node on Z-ICs with $\bar{\alpha}=[1,1,1]$.}
\label{figure_Z_DMT_MleqNi}
\end{figure}

Fig. \ref{figure_Z_DMT_MleqNi} depicts the comparison of the achievable GDMT of the IIML decoder with that of the fundamental CSIT GDMT of the channel on two different MIMO Z-ICs. Comparing the performance improvement of the IIML scheme on the $(3,4,3,3)$ Z-IC with respect to that on the $(3,3,3,3)$ Z-IC, we see that a larger number of antennas at the interfered receiver can completely compensate for the lack of CSIT. Again by the argument that the CSIT GDMT represents an upper bound to the No-CSIT GDMT of the channel, the observation from Fig.~\ref{figure_Z_DMT_MleqNi} implies that the GDMT of Lemma~\ref{lem:DMT-IML-asymmetric} represents the fundamental GDMT of the $(3,4,3,3)$ Z-IC with no CSIT. It turns out that, this channel is only a member of a large class of Z-ICs for which the No-CSIT DMT can be characterized. This class of channels is specified in the next theorem.

\begin{thm}
\label{thm:larger_antenna}
The DMT specified in Lemma~\ref{lem:DMT-IML-asymmetric} represents the fundamental DMT of the channel with No-CSIT and $r_1=r_2=r$, if
\begin{equation*}
N_1\geq M+\frac{d_{M,\min\{N_1,N_2\}}\left(\frac{M}{2}\right)}{M}.
\end{equation*}
\end{thm}
\begin{proof}
Comparing Theorem~\ref{thm:asymmetric-FCSIT} and Lemma~\ref{lem:DMT-IML-asymmetric}, the desired result can be obtained following the similar steps as in the proof of Theorem~\ref{thm:DMT-symmetric-no-CSIT}.
\end{proof}

\section{Conclusion}
\label{sec:conclusion}
The GDMT of the MIMO Z-IC with CSIT is characterized. It is shown that the knowledge of $H_{21}$ at the second transmitter only is sufficient to achieve the CSIT GDMT of the channel. The No-CSIT GDMT of two special class of Z-ICs have been characterized revealing useful insights about the system such as that a stronger interference or a larger number of antennas at the interfered receiver can completely compensate for the lack of CSIT
on a Z-IC. A complete characterization of the GDMT of the general Z-IC with no CSIT remains open.

\appendices
\section{Proof of Lemma~\ref{lem:ZIC-optimization-problem}}
\label{pf:lem:ZIC-optimization-problem}
Recall that, the negative SNR exponents of $W_1,~W_2$ and $W_3$ are denoted by $\{\beta_j\}_{j=1}^{q_1}$, $\{\gamma_k\}_{k=1}^{q_2}$ and $\{\upsilon_i\}_{i=1}^{p}$, respectively. Interestingly, these matrices have exactly the same structure specified in Theorem~1 and 2 of \cite{DDF-DMT-Journal-IT}, except from the presence of $\alpha_{21}$ in $\rho_{21} =\rho^{\alpha_{21}}$ which can be easily taken care of. Therefore, using these theorems we obtain the following conditional distributions:

\begin{IEEEeqnarray}{rl}
\label{eq:pdf-conditional-13}
f_{W_1|W_3}(\bar{\beta}|\bar{\upsilon})\dot{=}\left\{\begin{array}{cc}
\rho^{-E_1(\bar{\beta},\bar{\upsilon})} & \textrm{if }~(\bar{\beta},\bar{\upsilon})\in \mathcal{B}_1;\\
0& \textrm{if }~(\bar{\beta},\bar{\upsilon})\notin \mathcal{B}_1
\end{array}\right.
\end{IEEEeqnarray}
where
\begin{IEEEeqnarray}{rl}
E_1(\bar{\beta},\bar{\upsilon})= & \left\{\sum_{j=1}^{q_1}\left((M_1+N_1+1-2j)\beta_j+\sum_{i=1}^{\min\{(N_1-j), M_1\}}(\alpha_{21}-\upsilon_i-\beta_j)^+\right)-M_1\sum_{i=1}^{p}(\alpha_{21}-\upsilon_i)^+\right\},\\
\mathcal{B}_1= & \Big\{(\bar{\beta},\bar{\upsilon}):~  0\leq \upsilon_1\leq \cdots \leq \upsilon_{p};~
 0\leq \beta_1\leq \cdots \leq \beta_{q_1};~
 (\upsilon_i+\beta_j)\geq \alpha_{21}, ~\forall (i+j)\geq (N_1+1)\Big\}
\end{IEEEeqnarray}
and
\begin{IEEEeqnarray}{rl}
\label{eq:pdf-conditional-23}
f_{W_2|W_3}(\bar{\gamma}|\bar{\upsilon})\dot{=}\left\{\begin{array}{cc}
\rho^{-E_2(\bar{\gamma},\bar{\upsilon})} & \textrm{if }~(\bar{\gamma},\bar{\upsilon})\in \mathcal{B}_2;\\
0& \textrm{if }~(\bar{\gamma},\bar{\upsilon})\notin \mathcal{B}_2
\end{array}\right.
\end{IEEEeqnarray}
where
\begin{IEEEeqnarray}{rl}
E_2(\bar{\gamma},\bar{\upsilon})= & \left\{\sum_{k=1}^{q_2}\left((M_2+N_2+1-2k)\gamma_k+\sum_{i=1}^{\min\{(M_2-k), N_2\}}(\alpha_{21}-\upsilon_i-\gamma_k)^+\right)-N_2\sum_{i=1}^{p}(\alpha_{21}-\upsilon_i)^+\right\},\\
\mathcal{B}_2= & \Big\{(\bar{\beta},\bar{\upsilon}):~  0\leq \upsilon_1\leq \cdots \leq \upsilon_{p};~
 0\leq \gamma_1\leq \cdots \leq \gamma_{q_2};~
 (\upsilon_i+\gamma_k)\geq \alpha_{21}, ~\forall (i+k)\geq (M_2+1)\Big\}
\end{IEEEeqnarray}
These distributions are easily obtained from equation (10) of \cite{DDF-DMT-Journal-IT} by changing it properly due to the presence of the $\rho^{\alpha_{21}}$ in place of $\rho$ in \cite{DDF-DMT-Journal-IT}.

It was proved in \cite{HDRC-journal} that given the non-zero eigenvalues of $W_3$, the matrices $W_1$ and $W_2$ are conditionally independent. Since the non-zero eigenvalues of $H_{21}H_{21}^\dagger$ and $H_{21}^\dagger H_{21}$ are exactly the same for each realization of the matrix $H_{21}$, given $\bar{\upsilon}$, $\bar{\beta}$ and $\bar{\gamma}$ are conditionally independent of each other. Intuitively, it is a well known fact in the literature of random matrix theory that the eigenvalues of $W_1$ are dependent on the matrix $W_3$ only through its eigenvalues. The same is true for $W_2$. Therefore, given the eigenvalues of $W_3$, the eigenvalues of $W_1$ and $W_2$ are conditionally independent\footnote{For a detailed proof of this fact the reader is referred to Lemma~1 of \cite{HDRC-journal}.}. Using this fact we have the following:

\begin{IEEEeqnarray}{rl}
\label{eq:pdf-product-form}
f_{W_1,W_2,W_3}(\bar{\beta},\bar{\gamma},\bar{\upsilon})= & f_{W_1,W_2|W_3}(\bar{\beta},\bar{\gamma}|\bar{\upsilon}) f_{W_3}(\bar{\upsilon}),\\
\stackrel{(a)}{=} & f_{W_1|W_3}(\bar{\beta}|\bar{\upsilon})f_{W_2|W_3}(\bar{\gamma}|\bar{\upsilon}) f_{W_3}(\bar{\upsilon}).
\end{IEEEeqnarray}

Now, substituting equations \eqref{eq:pdf-conditional-13}, \eqref{eq:pdf-conditional-23} into the above equation we obtain the joint distribution $f_{W_1,W_2,W_3}(\bar{\beta},\bar{\gamma},\bar{\upsilon})$, where the marginal distributions of $\upsilon_i$'s are given by
\begin{IEEEeqnarray}{rl}
\label{eq:pdf-3}
f_{W_3}(\bar{\upsilon})\dot{=}\left\{\begin{array}{cc}
\rho^{-\sum_{i=1}^p(M_2+N_1+1-2i)\upsilon_i} & \textrm{if }~0\leq \upsilon_1\leq \cdots \leq \upsilon_p;\\
0& \textrm{otherwise},
\end{array}\right.
\end{IEEEeqnarray}
which was derived in \cite{tse1}.
Using this joint distribution $f_{W_1,W_2,W_3}(\bar{\beta},\bar{\gamma},\bar{\upsilon})$ and the Laplace's method and following the similar approach as in \cite{tse1} equation \eqref{eq:outage-sum-bound-FCSIT} can be evaluated to obtain the following:
\begin{subequations}
\label{eq:main-optimization-1}
\begin{align}
d_{\mathcal{O}_s}(r_s)= & \min \sum_{i=1}^{p}(M_2+N_1+1-2i)\upsilon_i +\sum_{j=1}^{q_1}(M_1+N_1+1-2j)\beta_j+\sum_{k=1}^{q_2}(M_2+N_2+1-2k)\gamma_k \nonumber \\
& -(M_1+N_2)\sum_{i=1}^p(\alpha_{21}-\upsilon_i)^+ +\sum_{k=1}^{q_2}\sum_{i=1}^{\min\{(M_2-k), N_2\}}(\alpha_{21}-\upsilon_i-\gamma_k)^+\nonumber \\
\label{eq:main-optimization-1a}
 & +\sum_{j=1}^{q_1}\sum_{i=1}^{\min\{(N_1-j), M_1\}}(\alpha_{21}-\upsilon_i-\beta_j)^+;\\
\label{eq:main-optimization-1b}
\textrm{constrained by: } &
\sum_{i=1}^{p}(\alpha_{21}- \upsilon_i)^++\sum_{j=1}^{q_1}(\alpha_{11}-\beta_j)^++\sum_{k=1}^{q_2}(\alpha_{22}-\gamma_k)^+< r_s;\\
\label{eq:main-optimization-1c}
& 0\leq \upsilon_1\leq \cdots \leq \upsilon_{p};\\
\label{eq:main-optimization-1d}
& 0\leq \beta_1\leq \cdots \leq \beta_{q_1};\\
\label{eq:main-optimization-1e}
& 0\leq \gamma_1\leq \cdots \leq \gamma_{q_2};\\
\label{eq:main-optimization-1f}
& (\upsilon_i+\beta_j)\geq \alpha_{21}, ~\forall (i+j)\geq (N_1+1);\\
\label{eq:main-optimization-1g}
& (\upsilon_i+\gamma_k)\geq \alpha_{21}, ~\forall (i+k)\geq (M_2+1).
\end{align}
\end{subequations}

Finally, the desired result follows from the fact that by restricting $\alpha_i\leq \alpha_{21}$ for all $i\leq p$ does not change the optimal solution of the above problem. This can be proved as follows: suppose the optimal solution has $\alpha_i>\alpha_{21}$ for some $i$. Now, since the objective function is monotonically decreasing function of $\alpha_i$ for all $i$, substituting $\alpha_i=\alpha_{21}$ does not violate any of the constraints but reduces the objective function. However, that means the earlier solution was not really the optimal solution. Therefore, in the optimal solution we must have $\alpha_i\leq \alpha_{21}$ for all $i\leq p$. However, with this constraint we have
\begin{equation*}
    (\alpha_{21}-\upsilon_i)^+=(\alpha_{21}-\upsilon_i),~ \forall~i\leq p.
\end{equation*}

Substituting this in equation \eqref{eq:main-optimization-1a} we get the desired result.

\section{Proof of Theorem~\ref{thm:DMT-symmetric-Z-FCSIT}}
\label{pf:thm:DMT-symmetric-Z-FCSIT}
The main steps of the proof can be described as follows. First, we simplify the optimization problem in equations \eqref{eq:main-optimization-a} and \eqref{eq:main-optimization} by putting specific values of the different parameters, such as $M_i$, $N_i$ and $\alpha_{ij}$ as stated in the statement of the theorem, in equations \eqref{eq:main-optimization-a} and \eqref{eq:main-optimization}. Then we calculate a local minimum of this simplified optimization problem for each value of $r$, using the steepest descent method, which by the previous argument then represents a global minimum. The latter part of the problem, i.e., the computation of the local minimum will be carried out in two steps: in step one, we consider the case when $\alpha\leq 1$ and then in the second step we consider the remaining case. Let us start by deriving the simplified optimization problem.

Substituting $M_1=M_2=N_1=N_2=n$, $\alpha_{11}=\alpha_{22}=1$ and $\alpha_{21}=\alpha$ in the optimization problem of Lemma~\ref{lem:ZIC-optimization-problem} we obtain the following:
\begin{subequations}
\label{eq:symmetric-optimization}
\begin{align}
\min \mathcal{O}_s \triangleq \min_{\left(\vec{\upsilon},\vec{\gamma},\vec{\beta}, \bar{\alpha} \right)} \sum_{i=1}^n(4n+1-2i)\upsilon_i +&\sum_{j=1}^n(2n+1-2j)\beta_j+\sum_{k=1}^n(2n+1-2k)\gamma_k \nonumber \\
\label{eq:symmetric-optimization-a}
  -2n^2&\alpha +\sum_{k=1}^n\sum_{i=1}^{(n-k)}(\alpha-\upsilon_i-\gamma_k)^+
+\sum_{j=1}^n\sum_{i=1}^{(n-j)}(\alpha-\upsilon_i-\beta_j)^+;\\
\label{eq:symmetric-optimization-b}
\textrm{constrained by: } &
\sum_{i=1}^n(\alpha- \upsilon_i)^++\sum_{j=1}^n(1-\beta_j)^++\sum_{k=1}^n(1-\gamma_k)^+< r_s;\\
\label{eq:symmetric-optimization-c}
& 0\leq \upsilon_1\leq \cdots \leq \upsilon_n;\\
\label{eq:symmetric-optimization-d}
& 0\leq \beta_1\leq \cdots \leq \beta_n;\\
\label{eq:symmetric-optimization-e}
& 0\leq \gamma_1\leq \cdots \leq \gamma_n;\\
\label{eq:symmetric-optimization-f}
& (\upsilon_i+\beta_j)\geq \alpha, ~\forall (i+j)\geq (n+1);\\
\label{eq:symmetric-optimization-g}
& (\upsilon_i+\gamma_k)\geq \alpha, ~\forall (i+k)\geq (n+1).
\end{align}
\end{subequations}

As stated earlier, in what follows, we solve this optimization problem in two steps; in the first step, we assume $\alpha\leq 1$.

\subsection{Step 1: ($\alpha\leq 1$)}
To apply the steepest descent method we first compute the rate of change of the objective function with respect to the various parameters. Differentiating the objective function in equation \eqref{eq:symmetric-optimization-a} we obtain the following:
\begin{IEEEeqnarray}{l}
\left.\frac{\partial \mathcal{O}_s}{\partial \alpha_i}\right|_{\alpha_1=0,\cdots, \alpha_{i-1}=0,\alpha_{i+1}=1,\cdots,\alpha_n=1, \beta_1=1,\gamma_1=1}=(4n+1-2i), ~1\leq i\leq n;\\
\left.\frac{\partial \mathcal{O}_s}{\partial \beta_1}\right|_{\alpha_1=1, \gamma_1=1}=(2n-1)\leq (4n+1-2i), ~\forall~1\leq i\leq n;\\
\left.\frac{\partial \mathcal{O}_s}{\partial \gamma_1}\right|_{\alpha_1=1, \beta_1=1}=(2n-1)\leq (4n+1-2i), ~\forall~1\leq i\leq n.
\end{IEEEeqnarray}
Note that it is sufficient, to consider the decay of the function with respect to (w.r.t.) $\beta_1$ and $\gamma_1$ only, because of the decreasing slope of the objective function with increasing index of $\beta$ and $\gamma$ and equation \eqref{eq:symmetric-optimization-d} and \eqref{eq:symmetric-optimization-e}. Therefore, it is clear from the slopes of the function given above that for $(i-1)\alpha\leq r_s\leq i\alpha$, the steepest descent of the function is along decreasing values of $\alpha_i$, while $\beta_1=\gamma_1=1$. Note that for these values of $\beta_1$ and $\gamma_1$ the last two terms of equation \eqref{eq:symmetric-optimization-a} vanishes and equations \eqref{eq:symmetric-optimization-d}-\eqref{eq:symmetric-optimization-g} becomes redundant and as a result the optimization problem of  \eqref{eq:symmetric-optimization} simplifies to the following:
\begin{subequations}
\label{eq:symmetric-optimization1}
\begin{align}
\min &\sum_{i=1}^n(4n+1-2i)\upsilon_i +2n^2(1-\alpha);\\
\label{eq:symmetric-optimization1-b}
\textrm{constrained by: } &
\sum_{i=1}^n(\alpha- \upsilon_i)^+\leq  r_s;\\
\label{eq:symmetric-optimization1-c}
& 0\leq \upsilon_1\leq \cdots \leq \upsilon_n.
\end{align}
\end{subequations}
The solution of this optimization problem follows from Lemma~\ref{lem:m-dmt} and is given by
\begin{IEEEeqnarray}{l}
\label{eq:dmt-symmetric-1a}
d_s(r_s)=\alpha d_{n,3n}\left(\frac{r_s}{\alpha}\right)+2n^2(1-\alpha), ~0\leq r_s\leq n\alpha.
\end{IEEEeqnarray}

The above solution also implies that for $\alpha\geq n\alpha$, the optimal solution have $\alpha_i=0,~ \forall i$, which when substituted in equation \eqref{eq:symmetric-optimization} we obtain the following problem.
\begin{subequations}
\label{eq:symmetric-optimization2}
\begin{align}
\label{eq:symmetric-optimization2-a}
\min &\sum_{j=1}^n(2n+1-2j)\beta_j+\sum_{k=1}^n(2n+1-2k)\gamma_k -2n^2,\\
\label{eq:symmetric-optimization2-b}
\textrm{constrained by: } &
\sum_{j=1}^n(1-\beta_j)^++\sum_{k=1}^n(1-\gamma_k)^+< (r_s-n\alpha);\\
\label{eq:symmetric-optimization2-c}
& \alpha\leq \beta_1\leq \cdots \leq \beta_n;\\
\label{eq:symmetric-optimization2-d}
& \alpha\leq \gamma_1\leq \cdots \leq \gamma_n.
\end{align}
\end{subequations}
Note that the last two summands in the objective function \eqref{eq:symmetric-optimization} is zero because of equations \eqref{eq:symmetric-optimization2-c} and \eqref{eq:symmetric-optimization2-d}. Now, from the symmetry of the optimization problem \eqref{eq:symmetric-optimization2} w.r.t. $\beta_i$ and $\gamma_i$, we can assume without loss of generality that the optimal solution will have $\beta_i=\gamma_i,~\forall i$. Substituting this and $\delta_i=\beta_i-\alpha$ in equation \eqref{eq:symmetric-optimization2} we get the following equivalent optimization problem:
\begin{subequations}
\label{eq:symmetric-optimization3}
\begin{align}
\label{eq:symmetric-optimization3-a}
\min \; &2\sum_{j=1}^n(2n+1-2j)\delta_j,\\
\label{eq:symmetric-optimization3-b}
\textrm{constrained by: } &
\sum_{j=1}^n(1-\alpha-\delta_j)^+\leq \left(\frac{r_s-n\alpha}{2}\right);\\
\label{eq:symmetric-optimization3-c}
& 0\leq \delta_1\leq \cdots \leq \delta_n.
\end{align}
\end{subequations}
The solution of this optimization problem follows again from Lemma~\ref{lem:m-dmt} and is given by
\begin{IEEEeqnarray}{rl}
d_s(r_s)=& 2(1-\alpha) d_{n,n}\left(\frac{r_s-n\alpha}{2(1-\alpha)}\right), ~0\leq \frac{r_s-n\alpha}{2}\leq n(1-\alpha)\nonumber\\
\label{eq:dmt-symmetric-1b}
=&2(1-\alpha) d_{n,n}\left(\frac{r_s-n\alpha}{2(1-\alpha)}\right), ~n\alpha\leq r_s\leq n(2-\alpha).
\end{IEEEeqnarray}
Combining equations \eqref{eq:dmt-symmetric-1a} and \eqref{eq:dmt-symmetric-1b} we obtain equation \eqref{eq:dmt-symmetric-a-leq-1} of Theorem~\ref{thm:DMT-symmetric-Z-FCSIT} and we have completed the first step of this proof. In what follows we consider the remaining case, when $\alpha\geq 1$.

\subsection{Step 2: ($\alpha\geq 1$)}
Differentiating the objective function in equation \eqref{eq:symmetric-optimization-a} we obtain the following:
\begin{IEEEeqnarray}{l}
\label{eq:slope-sym-2a-alpha1}
\left.\frac{\partial \mathcal{O}_s}{\partial \alpha_1}\right|_{\alpha_i= \alpha~\forall~i\geq 2, \beta_1=1,\gamma_1=1}=\left\{\begin{array}{cc}(4n-1),& \textrm{for}~\alpha_1\geq (\alpha-1);\\
(2n-1),& \textrm{for}~0\leq \alpha_1\leq (\alpha-1);\end{array}\right.\\
\label{eq:slope-sym-2a-beta1-gamma1}
\left.\frac{\partial \mathcal{O}_s}{\partial \beta_1}\right|_{\alpha_i= \alpha~\forall~i\geq 2,\gamma_1=1}=\left.\frac{\partial \mathcal{O}_s}{\partial \gamma_1}\right|_{\alpha_i= \alpha~\forall~i\geq 2,\beta_1=1}=\left\{\begin{array}{cc}(2n-1),& \textrm{for}~\alpha_1\geq (\alpha-1);\\
(2n-2),& \textrm{for}~0\leq \alpha_1\leq (\alpha-1);\end{array}\right.\\
\label{eq:slope-sym-2a-alpha2}
\left.\frac{\partial \mathcal{O}_s}{\partial \alpha_1}\right|_{\alpha_1=(\alpha-1),\alpha_i= \alpha~\forall~i\geq 3, \beta_1=1,\gamma_1=1}=(4n-3), \textrm{for}~\alpha_2\geq (\alpha-1).
\end{IEEEeqnarray}
Comparing equations \eqref{eq:slope-sym-2a-alpha1} and \eqref{eq:slope-sym-2a-beta1-gamma1}, we realize that for $0\leq r_s\leq 1$, the steepest descent is along the direction of decreasing $\alpha_1$, while $\beta_1=\gamma_1=1$. On the other hand, comparing equations \eqref{eq:slope-sym-2a-alpha1}, \eqref{eq:slope-sym-2a-beta1-gamma1} and \eqref{eq:slope-sym-2a-alpha2} it is clear that beyond $r_s=1$, decreasing $\alpha_2$ has the steepest descent than $\beta_1, \gamma_1$ and even $\alpha_1$. In the same way it can be proved that for $(i-1)\leq r_s\leq i$, the steepest descent of the function is along decreasing values of $\alpha_i$, while $\beta_1=\gamma_1=1$. Note that for these values of $\beta_1$ and $\gamma_1$ the last two terms of equation \eqref{eq:symmetric-optimization-a} vanish and equations \eqref{eq:symmetric-optimization-d}-\eqref{eq:symmetric-optimization-g} becomes redundant and as a result the optimization problem of  \eqref{eq:symmetric-optimization} simplifies to the following:
\begin{subequations}
\label{eq:symmetric-optimization4}
\begin{align}
\min &\sum_{i=1}^n(4n+1-2i)\upsilon_i -2n^2(\alpha-1);\\
\label{eq:symmetric-optimization4-b}
\textrm{constrained by: } &
\sum_{i=1}^n(\alpha- \upsilon_i)^+\leq  r_s;\\
\label{eq:symmetric-optimization4-c}
& (\alpha-1)\leq \upsilon_1\leq \cdots \leq \upsilon_n.
\end{align}
\end{subequations}
By the aforementioned argument, each of the $\upsilon_i$'s can decrease with the constraint $\upsilon_i\geq (\alpha-1)$ with increasing $r_s$. Therefore, substituting $\upsilon_i'=\upsilon_i-(\alpha-1)$ for all $i$ in the above set of equations we obtain the following equivalent optimization problem:
\begin{subequations}
\label{eq:symmetric-optimization5}
\begin{align}
\min &\sum_{i=1}^n(4n+1-2i)\upsilon_i' +n^2(\alpha-1);\\
\label{eq:symmetric-optimization5-b}
\textrm{constrained by: } &
\sum_{i=1}^n(1- \upsilon_i')^+\leq  r_s;\\
\label{eq:symmetric-optimization5-c}
& 0\leq \upsilon_1'\leq \cdots \leq \upsilon_n',
\end{align}
\end{subequations}
which in turn by Lemma~\ref{lem:m-dmt} has the following optimal value
\begin{IEEEeqnarray}{l}
\label{eq:dmt-symmetric-2a}
d_s(r_s)=\alpha d_{n,3n}\left(r_s\right)+n^2(\alpha-1), ~0\leq r_s\leq n.
\end{IEEEeqnarray}

It is clear from this solution that for $r_s\geq n$, $\alpha_i\leq (\alpha-1)$ for all $i$, and the optimization problem \eqref{eq:symmetric-optimization} reduces to the following:
\begin{subequations}
\label{eq:symmetric-optimization6}
\begin{align}
\min \mathcal{O}_s \triangleq \min_{\left(\vec{\upsilon},\vec{\gamma},\vec{\beta}, \bar{\alpha} \right)} \sum_{i=1}^n(4n+1-2i)\upsilon_i +&\sum_{j=1}^n(2n+1-2j)\beta_j+\sum_{k=1}^n(2n+1-2k)\gamma_k \nonumber \\
\label{eq:symmetric-optimization6-a}
  -2n^2&\alpha +\sum_{k=1}^n\sum_{i=1}^{(n-k)}(\alpha-\upsilon_i-\gamma_k)^+
+\sum_{j=1}^n\sum_{i=1}^{(n-j)}(\alpha-\upsilon_i-\beta_j)^+;\\
\label{eq:symmetric-optimization6-b}
\textrm{constrained by: } &
\sum_{i=1}^n(\alpha-1- \upsilon_i)+\sum_{j=1}^n(1-\beta_j)^++\sum_{k=1}^n(1-\gamma_k)^+\leq (r_s-n);\\
\label{eq:symmetric-optimization6-c}
& 0\leq \upsilon_1\leq \cdots \leq \upsilon_n\leq (\alpha-1);\\
\label{eq:symmetric-optimization6-d}
& 0\leq \beta_1\leq \cdots \leq \beta_n;\\
\label{eq:symmetric-optimization6-e}
& 0\leq \gamma_1\leq \cdots \leq \gamma_n;\\
\label{eq:symmetric-optimization6-f}
& (\upsilon_i+\beta_j)\geq \alpha, ~\forall (i+j)\geq (n+1);\\
\label{eq:symmetric-optimization6-g}
& (\upsilon_i+\gamma_k)\geq \alpha, ~\forall (i+k)\geq (n+1).
\end{align}
\end{subequations}
Note the upper bound in equation \eqref{eq:symmetric-optimization6-c}, which is different from equation \eqref{eq:symmetric-optimization-c}. Also, since we are seeking for the minimum value of the objective function which decreases with every $\beta_i$ and $\gamma_j$, in the optimal solution these parameters must take their minimum value, which from equations \eqref{eq:symmetric-optimization6-f} and \eqref{eq:symmetric-optimization6-g} is given by
\begin{IEEEeqnarray}{l}
\beta_{j}|_{\min}=\gamma_{j}|_{\min}=\alpha-\upsilon_{n-j+1},
\end{IEEEeqnarray}
which along with the ordering among the $\upsilon_i$'s, $\beta_j$'s and $\gamma_j$'s imply that all the terms in the last two summands of the objective function are non-negative. Substituting these minimum values and thereby eliminating $\beta_j$'s and $\gamma_j$'s from the optimization problem we obtain the following equivalent optimization problem:

\begin{subequations}
\label{eq:symmetric-optimization7}
\begin{align}
\label{eq:symmetric-optimization7-a}
\min & \sum_{i=1}^n(2n+1-2i)\upsilon_i, \\
\label{eq:symmetric-optimization7-b}
\textrm{constrained by: } &
\sum_{i=1}^n(\alpha-1- \upsilon_i)\leq (r_s-n);\\
\label{eq:symmetric-optimization7-c}
& 0\leq \upsilon_1\leq \cdots \leq \upsilon_n\leq (\alpha-1),
\end{align}
\end{subequations}
where the simplification of the objective function involves the following algebraic computations
\begin{IEEEeqnarray*}{rl}
\mathcal{O}_s = &\sum_{i=1}^n(4n+1-2i)\upsilon_i +2 \sum_{j=1}^n(2n+1-2j)(\alpha-\upsilon_{n-j+1})-2n^2\alpha \nonumber \\
  &+\sum_{k=1}^n\sum_{i=1}^{(n-k)}(\alpha-\upsilon_i-(\alpha-\upsilon_{n-k+1}))
+\sum_{j=1}^n\sum_{i=1}^{(n-j)}(\alpha-\upsilon_i-(\alpha-\upsilon_{n-j+1}));\\
= &\sum_{i=1}^n(4n+1-2i)\upsilon_i +2 \sum_{j=1}^n(2n+1-2j)(\alpha-\upsilon_{n-j+1})-2n^2\alpha \nonumber \\
  &+n(n-1)\alpha-\sum_{i=1}^n2(n-i)\upsilon_i - \sum_{k=1}^n (\alpha-\upsilon_{n-k+1})
-\sum_{j=1}^n(n-j)(\alpha-\upsilon_{n-j+1});\\
= &\sum_{i=1}^n(2n+1)\upsilon_i +2 \sum_{j=1}^n(n+1-j)(\alpha-\upsilon_{n-j+1})-n(n+1)\alpha \nonumber\\
= &\sum_{i=1}^n(2n+1-2i)\upsilon_i +2 \sum_{j=1}^n(n+1-j)\alpha-n(n+1)\alpha \nonumber\\
=&\sum_{i=1}^n(2n+1-2i)\upsilon_i.
\end{IEEEeqnarray*}

The solution of this optimization problem follows again from Lemma~\ref{lem:m-dmt} and is given by
\begin{IEEEeqnarray}{rl}
d_s(r_s)=& (\alpha-1) d_{n,n}\left(\frac{r_s-n}{(\alpha-1)}\right), ~0\leq (r_s-n)\leq n(\alpha-1)\nonumber\\
\label{eq:dmt-symmetric-2b}
=&(\alpha-1) d_{n,n}\left(\frac{r_s-n}{(\alpha-1)}\right), ~n\leq r_s\leq n\alpha.
\end{IEEEeqnarray}
Combining equations \eqref{eq:dmt-symmetric-2a} and \eqref{eq:dmt-symmetric-2b} we obtain equation \eqref{eq:dmt-symmetric-a-geq-1} of Theorem~\ref{thm:DMT-symmetric-Z-FCSIT}.

\section{Proof of Theorem~\ref{thm:DMT-femto-symmetric}}
\label{pf:thm:DMT-femto-symmetric}
When we substitute the specific values of the different parameters in equation \eqref{eq:main-optimization}, it reduces to the following
\begin{subequations}
\label{eq:femto-optimization}
\begin{align}
d_{s,(n,n,n,n)}^{\textrm{Femto}}(r_s)=& \min \sum_{i=1}^n(4n+1-2i)\upsilon_i +\sum_{j=1}^n(2n+1-2j)\beta_j+\sum_{k=1}^n(2n+1-2k)\gamma_k \nonumber \\
\label{eq:femto-optimization-a}
 & -2n^2 +\sum_{k=1}^n\sum_{i=1}^{(n-k)}(1-\upsilon_i-\gamma_k)^+
+\sum_{j=1}^n\sum_{i=1}^{(n-j)}(1-\upsilon_i-\beta_j)^+;\\
\label{eq:femto-optimization-b}
\textrm{constrained by: } &
\sum_{i=1}^n(1- \upsilon_i)^++\sum_{j=1}^n(1-\beta_j)^++\sum_{k=1}^n(\alpha-\gamma_k)^+< r_s;\\
\label{eq:femto-optimization-c}
& 0\leq \upsilon_1\leq \cdots \leq \upsilon_n;\\
\label{eq:femto-optimization-d}
& 0\leq \beta_1\leq \cdots \leq \beta_n;\\
\label{eq:femto-optimization-e}
& 0\leq \gamma_1\leq \cdots \leq \gamma_n;\\
\label{eq:femto-optimization-f}
& (\upsilon_i+\beta_j)\geq 1, ~\forall (i+j)\geq (n+1);\\
\label{eq:femto-optimization-g}
& (\upsilon_i+\gamma_k)\geq 1, ~\forall (i+k)\geq (n+1).
\end{align}
\end{subequations}

Differentiating the objective function with respect to $\upsilon_i,~\forall~i, \beta_1$ and $\gamma_1$ we find that
\begin{IEEEeqnarray}{rl}
\left.\frac{\partial \mathcal{F}_{\textrm{Femto}}}{\partial \upsilon_i}\right|_{\beta_1=1,\gamma_1=\alpha}=& (4n+1-2i)\geq \left.\frac{\partial \mathcal{F}_{\textrm{Femto}}}{\partial \beta_1}\right|_{\gamma_1=\alpha, \upsilon_k=0,~\forall k<i, \upsilon_k=1,\forall k>i}=(2n-i),\\
=&\left.\frac{\partial \mathcal{F}_{\textrm{Femto}}}{\partial \gamma_1}\right|_{\beta_1=1, \upsilon_k=0,~\forall k<i, \upsilon_k=1,\forall k>i},~\forall ~i\leq n,
\end{IEEEeqnarray}
where we have denoted the objective function by $\mathcal{F}_{\textrm{Femto}}$. It is clear form these values that, for $(i-1)\leq r_s\leq i$, the steepest descent is along decreasing $\upsilon_i$ with $\beta_1=1$ and $\gamma_1=\alpha$. Putting this in equation \eqref{eq:femto-optimization} we get
\begin{subequations}
\label{eq:femto-optimization1}
\begin{align}
\label{eq:femto-optimization1-a}
d_{s,(n,n,n,n)}^{\textrm{Femto}}(r_s)=& \min \sum_{i=1}^n(4n+1-2i)\upsilon_i
+n^2(\alpha-1);\\
\label{eq:femto-optimization1-b}
\textrm{constrained by: } &
\sum_{i=1}^n(1- \upsilon_i)^+\leq r_s;\\
\label{eq:femto-optimization1-c}
& 0\leq \upsilon_1\leq \cdots \leq \upsilon_n.
\end{align}
\end{subequations}
Now, using Lemma~\ref{lem:m-dmt} we obtain the minimum value of the the above optimization problem, which can be written as
\begin{IEEEeqnarray}{l}
\label{eq:dmt-femto-2a}
d_{s,(n,n,n,n)}^{\textrm{Femto}}(r_s)=d_{n,3n}(r_s)+n^2(\alpha-1),~0\leq r_s\leq n.
\end{IEEEeqnarray}

Next, we obtain the optimal value of the objective function for values of $r_s\geq n$. Note that for any $r_s\geq n$, $\upsilon_i=0~\forall ~i$, which along with equations \eqref{eq:femto-optimization-f} and \eqref{eq:femto-optimization-g} imply that $\beta_j\geq 1,$ and $\gamma_j\geq 1~\forall j$. Putting this in equation \eqref{eq:femto-optimization} we get
\begin{subequations}
\label{eq:femto-optimization2}
\begin{align}
\label{eq:femto-optimization2-a}
d_{s,(n,n,n,n)}^{\textrm{Femto}}(r_s)=& \min \sum_{i=1}^n(2n+1-2i)\gamma_i-n^2
,\\
\label{eq:femto-optimization2-b}
\textrm{constrained by:} &
\sum_{i=1}^n(\alpha- \gamma_i)^+\leq (r_s-n);\\
\label{eq:femto-optimization2-c}
& 1\leq \gamma_1\leq \cdots \leq \gamma_n.
\end{align}
\end{subequations}
To bring the above problem into a form amenable to Lemma~\ref{lem:m-dmt} we use the following variable transform in the above set of equations: $\gamma_i'=\gamma_i-1$. This results in the following equivalent optimization problem.
\begin{subequations}
\label{eq:femto-optimization3}
\begin{align}
\label{eq:femto-optimization3-a}
d_{s,(n,n,n,n)}^{\textrm{Femto}}(r_s)=& \min \sum_{i=1}^n(2n+1-2i)\gamma_i'
,\\
\label{eq:femto-optimization3-b}
\textrm{constrained by:} &
\sum_{i=1}^n(\alpha-1- \gamma_i')^+\leq (r_s-n);\\
\label{eq:femto-optimization3-c}
& 0\leq \gamma_1'\leq \cdots \leq \gamma_n',
\end{align}
\end{subequations}
which in turn by Lemma~\ref{lem:m-dmt} attains the following optimal value:
\begin{IEEEeqnarray}{rl}
\label{eq:dmt-femto-2b}
d_{s,(n,n,n,n)}^{\textrm{Femto}}(r_s)=& (\alpha-1)d_{n,n}\left(\frac{(r_s-n)}{(\alpha-1)}\right),~0\leq \frac{(r_s-n)}{(\alpha-1)}\leq n\\
=& (\alpha-1)d_{n,n}\left(\frac{(r_s-n)}{(\alpha-1)}\right),~n\leq r_s\leq n\alpha.
\end{IEEEeqnarray}

Finally, combining equations \eqref{eq:dmt-femto-2a} and \eqref{eq:dmt-femto-2b} we obtain the desired result.

\section{Proof of Theorem~\ref{thm:asymmetric-FCSIT}}
\label{pf:thm:asymmetric-FCSIT}
Substituting $M_1=M_2=M$ and $\alpha_{ij}=1$ in equation \eqref{eq:main-optimization} we obtain
\begin{subequations}
\label{eq:asymmetric-FCSIT-optimization}
\begin{align}
d_{s,(M,N_1,M,N_2)}^{\textrm{FCSIT}}(r_s)= & \sum_{i=1}^M(2M+N_1+N_2+1-2i)\upsilon_i +\sum_{j=1}^M(M+N_1+1-2j)\beta_j+\sum_{k=1}^M(M+N_2+1-2k)\gamma_k \nonumber \\
\label{eq:asymmetric-FCSIT-optimization-a}
  & -(M+N_2)M +\sum_{k=1}^M\sum_{i=1}^{(M-k)}(1-\upsilon_i-\gamma_k)^+
+\sum_{j=1}^M\sum_{i=1}^{\min\{(N_1-j), M\}}(1-\upsilon_i-\beta_j)^+;\\
\label{eq:asymmetric-FCSIT-optimization-b}
\textrm{constrained by: } &
\sum_{i=1}^{p}(1- \upsilon_i)^++\sum_{j=1}^{q_1}(1-\beta_j)^++\sum_{k=1}^{q_2}(1-\gamma_k)^+< r_s;\\
\label{eq:asymmetric-FCSIT-optimization-c}
& 0\leq \upsilon_1\leq \cdots \leq \upsilon_{p};\\
\label{eq:asymmetric-FCSIT-optimization-d}
& 0\leq \beta_1\leq \cdots \leq \beta_{q_1};\\
\label{eq:asymmetric-FCSIT-optimization-e}
& 0\leq \gamma_1\leq \cdots \leq \gamma_{q_2};\\
\label{eq:asymmetric-FCSIT-optimization-f}
& (\upsilon_i+\beta_j)\geq 1, ~\forall (i+j)\geq (N_1+1);\\
\label{eq:asymmetric-FCSIT-optimization-g}
& (\upsilon_i+\gamma_k)\geq 1, ~\forall (i+k)\geq (M+1).
\end{align}
\end{subequations}

Differentiating the objective function with respect to $\upsilon_i,~\forall~i, \beta_1$ and $\gamma_1$ we find that
\begin{IEEEeqnarray}{rl}
\left.\frac{\partial \mathcal{F}_{\textrm{a-FCSIT}}}{\partial \upsilon_i}\right|_{\beta_1=1,\gamma_1=1}\geq & \left\{\begin{array}{c} \left.\frac{\partial \mathcal{F}_{\textrm{a-FCSIT}}}{\partial \beta_1}\right|_{\gamma_1=1, \alpha_k=0,~\forall k<i, \alpha_k=1,\forall k>i},\\
\left.\frac{\partial \mathcal{F}_{\textrm{a-FCSIT}}}{\partial \gamma_1}\right|_{\beta_1=1, \alpha_k=0,~\forall k<i, \alpha_k=1,\forall k>i},\end{array}\right.~\forall ~i\leq n,
\end{IEEEeqnarray}
where we have denoted the objective function by $\mathcal{F}_{\textrm{a-FCSIT}}$. It is clear form these values that, for $(i-1)\leq r_s\leq i$, the steepest descent is along decreasing $\upsilon_i$ with $\beta_1=\gamma_1=1$ and $i\leq n$. Putting this in equation \eqref{eq:asymmetric-FCSIT-optimization} we get
\begin{subequations}
\label{eq:asymmetric-FCSIT-optimization1}
\begin{align}
\label{eq:asymmetric-FCSIT-optimization1-a}
d_{s,(M,N_1,M,N_2)}^{\textrm{FCSIT}}(r_s)=& \min \sum_{i=1}^n(2M+N_1+N_2+1-2i)\upsilon_i
+M(N_1-M);\\
\label{eq:asymmetric-FCSIT-optimization1-b}
\textrm{constrained by: } &
\sum_{i=1}^M(1- \upsilon_i)^+\leq r_s;\\
\label{eq:asymmetric-FCSIT-optimization1-c}
& 0\leq \upsilon_1\leq \cdots \leq \upsilon_n.
\end{align}
\end{subequations}
Now, using Lemma~\ref{lem:m-dmt} we obtain the minimum value of the the above optimization problem, which can be written as
\begin{IEEEeqnarray}{l}
\label{eq:dmt-asymmetric-FCSIT-a}
d_{s,(M,N_1,M,N_2)}^{\textrm{FCSIT}}(r_s)=d_{M,M+N_1+N_2}(r_s)+M(N_1-M),~0\leq r_s\leq M.
\end{IEEEeqnarray}

Next, we evaluate the optimal value of the objective function for values of $r_s\geq M$. Note that for any $r_s\geq M$, $\upsilon_i=0~\forall ~i$, which along with equations \eqref{eq:asymmetric-FCSIT-optimization-f} and \eqref{eq:asymmetric-FCSIT-optimization-g} imply that $\beta_j\geq 1$ for $j\geq (N_1+1-M)$ and $\gamma_k\geq 1~\forall k$. Clearly, the objective function is minimized for $\beta_j=1$ for $j\geq (N_1+1-M)$ and $\gamma_k= 1~\forall k$. Putting this in equation \eqref{eq:asymmetric-FCSIT-optimization} we get
\begin{subequations}
\label{eq:asymmetric-FCSIT-optimization2}
\begin{align}
\label{eq:asymmetric-FCSIT-optimization2-a}
d_{s,(M,N_1,M,N_2)}^{\textrm{FCSIT}}(r_s)=& \min \sum_{j=1}^{M}(M+N_1+1-2j)\beta_j-M^2+\sum_{j=1}^M \min\{(N_1-j),M\}(1-\beta_j)^+
,\\
\label{eq:asymmetric-FCSIT-optimization2-b}
\textrm{constrained by:} &
\sum_{i=1}^M(1- \beta_j)\leq (r_s-M);\\
\label{eq:asymmetric-FCSIT-optimization2-c}
& 0\leq \beta_1\leq \cdots \leq \beta_M,
\end{align}
\end{subequations}
where $\beta_u=1$ for $u\geq \min\{(N_1-M),M\}$. Since $\beta_u=1$ only for $u\geq \min\{(N_1-M),M\}$ the last term in equation \eqref{eq:asymmetric-FCSIT-optimization2-a} reduces to
\begin{equation*}
\sum_{j=1}^M \min\{(N_1-j),M\}(1-\beta_j)^+=\sum_{j=1}^{\tau}M(1-\beta_j),
\end{equation*}
where we denote $\min\{(N_1-M),M\}=\tau$. When we substitute this, the last optimization problem further reduces to the following:
\begin{subequations}
\label{eq:asymmetric-FCSIT-optimization3}
\begin{align}
\label{eq:asymmetric-FCSIT-optimization3-a}
d_{s,(M,N_1,M,N_2)}^{\textrm{FCSIT}}(r_s)=& \min \sum_{j=1}^{\tau}(N_1+1-2j)\beta_j+(M-\tau)(N_1-\tau)-M^2+ \tau M
,\\
=& \min \sum_{j=1}^{\tau}(N_1+1-2j)\beta_j+(M-\tau)(N_1-M-\tau),\nonumber\\
=& \min \sum_{j=1}^{\tau}(\max\{N_1-M,M\}+\tau+1-2j)\beta_j,\nonumber\\
\label{eq:asymmetric-FCSIT-optimization3-b}
\textrm{constrained by:} &
\sum_{i=1}^{\tau}(1- \beta_j)\leq (r_s-M);\\
\label{eq:asymmetric-FCSIT-optimization3-c}
& 0\leq \beta_1\leq \cdots \leq \beta_{\tau}.
\end{align}
\end{subequations}
Using Lemma~\ref{lem:m-dmt} in the above optimization then it yields the following solution:
\begin{IEEEeqnarray}{rl}
d_{s,(M,N_1,M,N_2)}^{\textrm{FCSIT}}(r_s)=& d_{\tau,\max\{N_1-M,M\}}(r_s-M),~0\leq r_s-M\leq \tau;\nonumber\\
\label{eq:dmt-asymmetric-FCSIT-b}
=& d_{2M,N_1}(r_s),~M\leq r_s\leq \min\{N_1,2M\}.
\end{IEEEeqnarray}
Finally, combining equations \eqref{eq:dmt-asymmetric-FCSIT-a} and \eqref{eq:dmt-asymmetric-FCSIT-b} we obtain the desired result.

\section{Proof of Lemma~\ref{lem:IML-dmt-exponents}}
\label{pf:lem:IML-dmt-exponents}
Let us denote the event that a target rate tuple $(r_1\log(\rho),r_2\log(\rho))$ does not belong to $\mathcal{R}_{\textrm{IIML}}$ by $\mathcal{O}_{\textrm{IIML}}$, i.e.,
\begin{eqnarray*}
\mathcal{O}_{\textrm{IIML}}=\Big\{\mathcal{H}:~(r_1\log(\rho),r_2\log(\rho))\notin \mathcal{R}_{\textrm{IIML}}\Big\}.
\end{eqnarray*}
Now, let us denote the maximum among the average probability of errors at both the receivers be denoted by $\mathcal{P}_e$, then using Bayes' rule we get
\begin{IEEEeqnarray}{rl}
\mathcal{P}_e= & \mathcal{P}_{e|\mathcal{O}_{\textrm{IIML}}}\Pr\{\mathcal{O}_{\textrm{IIML}}\}+\mathcal{P}_{e|\mathcal{O}^c_{\textrm{IIML}}}\Pr\{\mathcal{O}^c_{\textrm{IIML}}\},\\
\leq & \Pr\{\mathcal{O}_{\textrm{IIML}}\}+\mathcal{P}_{e|\mathcal{O}^c_{\textrm{IIML}}},
\end{IEEEeqnarray}
where $\mathcal{P}_{e|\mathcal{E}}$ denote the conditional average probability of error given the event $\mathcal{E}$. Note that, the above equation holds for any SNR. When the target rate tuple belongs to $\mathcal{R}_{\textrm{IIML}}$, letting the block length be sufficiently large the probability of error given $\mathcal{O}^c_{\textrm{IIML}}$ can be be made arbitrarily close to zero. Therefore, letting the block length of the code goes to infinity at both side of the above equation we obtain
\begin{IEEEeqnarray}{rl}
\mathcal{P}_e~\leq  \Pr\{\mathcal{O}_{\textrm{IIML}}\}=&\Pr\left\{\{I_{c_1}\leq r_1\log(\rho)\}\cup \{I_{c_2}\leq r_2\log(\rho)\}\cup \{I_{c_s}\leq r_s\log(\rho)\}\right\},\\
\leq &\sum_{i=1,2,s}\Pr\left\{I_{c_i}\leq r_i\log(\rho)\right\},\\
\stackrel{(a)}{\dot{=}}&\max_{i=1,2,s}\Pr\left\{I_{c_i}\leq r_i\log(\rho)\right\},\\
\dot{=}&\max_{i=1,2,s}\rho^{-d_{i,(M_1,N_1,M_2,N_2)}^{\textrm{IIML}}(r_i)}=\rho^{-\min_{i=1,2,s} \{d_{i,(M_1,N_1,M_2,N_2)}^{\textrm{IIML}}(r_i)\}},
\end{IEEEeqnarray}
where step $(a)$ follows from the fact that in the asymptotic SNR the largest term dominates and the last step follows from equation \eqref{eq:def-outage-exponent-IML-decoder}. Finally, the desired result follows from the fact that $\mathcal{P}_e\dot{=}\rho^{-d_{(M,N_1,M,N_2)}^{\textrm{IIML}}(r_1,r_2)}$.

\section{Proof of Lemma~\ref{lem:IML-sum-optimization}}
\label{pf:lem:IML-sum-optimization}
The joint distribution of $(\bar{\beta},\bar{\upsilon})$ can be obtained from equation \eqref{eq:pdf-product-form} substituting $f_{W_2|W_3}(.)=1$. The rest of the proof follows the same steps as the proof of Lemma~\ref{lem:ZIC-optimization-problem} and hence skipped to avoid repeating.

\section{Proof of Theorem~\ref{thm:DMT-symmetric-no-CSIT}}
\label{pf:thm:DMT-symmetric-no-CSIT}
Since the CSIT DMT is an upper bound to the No-CSIT DMT, it is sufficient to prove that when the condition of equation \eqref{eq:condition-noCSIT-1} is satisfied the expressions of Theorem~\ref{thm:DMT-symmetric-Z-FCSIT} and Lemma~\ref{lem:DMT-symmetric-Z-achievable} are identical.

It is clear from the comparison of equations \eqref{eq:dmt-symmetric-a-geq-1} and \eqref{eq:dmt-IML-symmetric} that for $r_s=2r\geq n$ they are identical for all values of $\alpha\geq 1$ and hence when $\alpha$ satisfies equation \eqref{eq:condition-noCSIT-1}. Now, it is only necessary to find a condition when the expressions in equations \eqref{eq:dmt-symmetric-a-geq-1} and \eqref{eq:dmt-IML-symmetric} are identical even when $r_s<n$, which is what we derive next and turns out to be identical to the condition of equation \eqref{eq:condition-noCSIT-1}.

It is clear from equation \eqref{eq:dmt-IML-symmetric} and \eqref{eq:dmt-symmetric-a-geq-1} that,
\begin{equation*}
d_{s,(n,n,n,n)}^{\textrm{IIML}}(r_s)< d_{s,(n,n,n,n)}^{\textrm{CSIT}}(r_s),~\textrm{for}~ r_s\leq n.
\end{equation*}
Therefore, the DMTs given by Theorem~\ref{thm:DMT-symmetric-Z-FCSIT} and Lemma~\ref{lem:DMT-symmetric-Z-achievable} are identical only if for $r_s\leq n$, the single user performance is dominating, i.e.,
\begin{IEEEeqnarray}{rl}
d_{n,n}(r)\leq &d_{s,(n,n,n,n)}^{\textrm{IIML}}(r_s=2r);\nonumber \\
\label{eq:condition-sym-dmt-no-CSIT}
d_{n,n}(r)\leq & d_{n,2n}(2r)+n^2(\alpha-1),
\end{IEEEeqnarray}
where in the last step we have substituted the value of $d_{s,(n,n,n,n)}^{\textrm{IIML}}(2r)$ from equation \eqref{eq:dmt-IML-symmetric}. Since $d_{s,(n,n,n,n)}^{\textrm{IIML}}(2r)$ decays much faster than $d_{n,n}(r)$ with increasing $r$ and both are continuous functions of $r$, equation \eqref{eq:condition-sym-dmt-no-CSIT} will be valid for all $r\leq \frac{n}{2}$ if it is true for $r= \frac{n}{2}$. Substituting this in equation \eqref{eq:condition-sym-dmt-no-CSIT} we get
\begin{IEEEeqnarray*}{rl}
d_{n,n}\left(\frac{n}{2}\right)\leq & n^2(\alpha-1),\\
\textrm{or,}~\alpha \geq &1+\frac{d_{n,n}(\frac{n}{2})}{n^2}.
\end{IEEEeqnarray*}

\bibliographystyle{IEEETran}
\bibliography{mybibliography}

\end{document}